%% file: Journal-ModelFree.tex
\documentclass[journal,12pt,draftclsnofoot, onecolumn]{IEEEtran}
\IEEEoverridecommandlockouts

\usepackage{ifpdf}
\usepackage{cite}
\usepackage{placeins}
\usepackage[pdftex]{graphicx}
\usepackage{array}
\usepackage{caption}
\usepackage{subcaption}
\usepackage{fixltx2e}
\usepackage{stfloats}
\usepackage{hyperref}
\hypersetup{pdfborder = {0 0 0}}% for removing boxes surrounded hyperlink
\usepackage{amsthm}
\usepackage{amsmath}
\usepackage{amsfonts}
\usepackage{amssymb}
\usepackage[T1]{fontenc} % optional
\usepackage[cmintegrals]{newtxmath}
\usepackage{bm} % optional
\usepackage[all]{xy}
\usepackage{enumitem}
\usepackage{float}
\usepackage[usenames, dvipsnames]{color}
\usepackage{ifpdf}
\usepackage{cite}
\usepackage{array}
\usepackage{mathtools}

\makeatletter
\def\widebreve{\mathpalette\wide@breve}
\def\wide@breve#1#2{\sbox\z@{$#1#2$}%
	\mathop{\vbox{\m@th\ialign{##\crcr
				\kern0.08em\brevefill#1{0.8\wd\z@}\crcr\noalign{\nointerlineskip}%
				$\hss#1#2\hss$\crcr}}}\limits}
\def\brevefill#1#2{$\m@th\sbox\tw@{$#1($}%
	\hss\resizebox{#2}{\wd\tw@}{\rotatebox[origin=c]{90}{\upshape(}}\hss$}
\makeatletter

\usepackage{url}
\usepackage{mathtools}
\usepackage[dvipsnames]{xcolor}
\usepackage{multirow}
\usepackage{mwe} % loads »blindtext« and »graphicx«
\usepackage{times}
\usepackage{wasysym}
\usepackage{epsfig}
\usepackage{multirow}
\usepackage{tabularx}
\usepackage{graphicx}
\usepackage{color}
\usepackage{xspace}
\usepackage{thumbpdf}
\usepackage{listings}
\usepackage{verbatim}
\usepackage{outlines}
\usepackage{textgreek}
\usepackage{booktabs}
\usepackage{amsfonts}
\usepackage{colortbl}
\usepackage{multicol}
\usepackage{multirow}
\usepackage{balance}
\usepackage{makecell}
\usepackage{array}
\usepackage[normalem]{ulem}
\usepackage{mathtools,amssymb}
\usepackage{float} % set figures to be here exactly
\usepackage{booktabs} % For formal tables
\usepackage{wrapfig}
\usepackage{subfloat}

\usepackage{diagbox}

\theoremstyle{definition}

\newtheorem{thm}{Theorem}

\newtheorem{lem}{Lemma}

\newtheorem{define}{Definition}

\newcommand{\no}{\nonumber}

\usepackage[noend]{algpseudocode}

\usepackage{diagbox}
\theoremstyle{definition}

\DeclareMathOperator*{\argmax}{arg\,max}

\begin{document}
\title{Superstring-Based Sequence Obfuscation to Thwart Pattern Matching Attacks\\}

%\markboth{IEEE TRANSACTIONS ON WIRELESS COMMUNICATIONS}

\author{Bo~Guan~\IEEEmembership{Student Member,~IEEE,}
	Nazanin~Takbiri~\IEEEmembership{Student Member,~IEEE,}
	Dennis~Goeckel~\IEEEmembership{Fellow,~IEEE,}\\
	Amir~Houmansadr~\IEEEmembership{Member,~IEEE,}
	Hossein~Pishro-Nik~\IEEEmembership{Member,~IEEE}

	\thanks{B. Guan, N. Takbiri, D. L. Goeckel, and H. Pishro-Nik are with the Department
		of Electrical and Computer Engineering, University of Massachusetts, Amherst,
		MA, 01003 USA. e-mail: \{boguan, ntakbiri, goeckel, pishro\}@ecs.umass.edu}% <-this % stops a space
	\thanks{A. Houmansadr is with the College of Information and Computer Sciences, University of Massachusetts, Amherst, MA, 01003 USA. e-mail:(amir@cs.umass.edu).}
	\thanks{This work was supported by the National Science Foundation under grants CCF--1421957 and CNS--1739462.}
	\thanks{This work was presented in part in IEEE International Symposium on Information Theory (ISIT 2020)~\cite{guan2020sequence}.}
	\thanks{This work has been submitted to the IEEE for possible publication. Copyright may be transferred without notice, after which this
version may no longer be accessible.}
}

%%% Many authors with many affiliations:
% \author{%
%   \IEEEauthorblockN{Bo~Guan\IEEEauthorrefmark{1},
%                     Nazanin~Takbiri\IEEEauthorrefmark{1},
%                     Dennis~Goeckel\IEEEauthorrefmark{1},
%                     Amir~Houmansadr\IEEEauthorrefmark{2},
%                     Hossein~Pishro-Nik\IEEEauthorrefmark{1},}
%   \IEEEauthorblockA{\IEEEauthorrefmark{1}%
%                     Department of Electrical and Computer Engineering,
%                     University of Massachusetts,
%                     Amherst, MA, 01003 USA\\
%                     \{boguan, ntakbiri, goeckel, pishro\}@ecs.umass.edu}
%   \IEEEauthorblockA{\IEEEauthorrefmark{2}%
%                     College of Information and Computer Sciences,
%                     University of Massachusetts,
%                     Amherst, MA, 01003 USA\\
%                     amir@cs.umass.edu}
% }
	%\thanks{This work was supported by National Science Foundation under grants CCF--1421957 and CNS--1739462.}
	
\maketitle

\begin{abstract}
% revised by Dennis for journal version
User privacy can be compromised by matching user data traces to records of their previous behavior. The matching of the statistical characteristics of traces to prior user behavior has been widely studied. However, an adversary can also identify a user deterministically by searching data traces for a pattern that is unique to that user. Our goal is to thwart such an adversary by applying small artificial distortions to data traces such that each potentially identifying pattern is shared by a large number of users. Importantly, in contrast to statistical approaches, we develop data-independent algorithms that require no assumptions on the model by which the traces are generated. By relating the problem to a set of combinatorial questions on sequence construction, we are able to provide provable guarantees for our proposed constructions. We also introduce data-dependent approaches for the same problem. The algorithms are evaluated on synthetic data traces and on the Reality Mining Dataset to demonstrate their utility.
% ISIT version
%Suppose we are given a large number of sequences on a given alphabet, and an adversary is interested in identifying (de-anonymizing) a specific target sequence based on its patterns. Our goal is to thwart such an adversary by obfuscating the target sequences by applying artificial (but small) distortions to its values. A key point here is that we would like to make no assumptions about the statistical model of such sequences. This is in contrast to existing literature where  assumptions (e.g., Markov chains) are made regarding such sequences to obtain privacy guarantees. We relate this problem to a set of combinatorial questions on sequence construction based on which we are able to obtain provable guarantees. This problem is relevant to important privacy applications: from fingerprinting webpages visited by users through anonymous communication systems to linking communicating parties on messaging applications to inferring activities of users of  IoT devices.
\end{abstract}

%\IEEEpeerreviewmaketitle
\input{introduction-1}
\input{relatedWork}
\input{model}

%%%%%%%%%%%%%%%%%%%%%%%%%%%%%%%%%%%%%%%%%%%%%
\input{data_independent}
%%%%%%%%%%%%%%%%%%%%%%%%%%%%%%%%%%%%%%%%%%%%%
\input{combination}
%%%%%%%%%%%%%%%%%%%%%%%%%%%%%%%%%%%%%%%%%%%%%
\input{data_dependent}
\input{complexity}
%%%%%%%%%%%%%%%%%%%%%%%%%%%%%%%%%%%%%%%%%%%%%
\input{numerical_results}
%%%%%%%%%%%%%%%%%%%%%%%%%%%%%%%%%%%%%%%%%%%%%
\input{conclusion}
\appendices
\FloatBarrier
\bibliographystyle{IEEEtran}
\bibliography{ref}
\end{document}

%% file: introduction-1.tex
\section{Introduction}
\label{intro}

The prominence of the Internet of Things (IoT) has raised security and privacy concerns.  The problem considered here addresses several important scenarios: fingerprinting webpages visited by users through anonymous communication systems \cite{shirani2018optimal} \cite{wondracek2010practical}, linking communicating parties on messaging applications \cite{danezis2009vida}, and inferring the activities of the users of IoT devices \cite{junges2019passive} \cite{apthorpe2017spying}.  While the setting is general, we motivate the problem from the consideration of \textit{User-Data Driven} (UDD) services in IoT applications: data submitted by users is analyzed to improve service in applications such as health care, smart homes, and connected vehicles.  But privacy and security threats are a major obstacle to the wide adoption of IoT applications~\cite{staff2015internet,sadeghi2015security,wang2017estimation, ukil2014iot, diaz2019robustness}.  
Often anonymization and obfuscation mechanisms are proposed to improve privacy at the cost of user utility.  Anonymization techniques frequently change the pseudonym of users~\cite{shirani2018optimal, hoh2005protecting, unnikrishnan2014asymptotically, keen, Naini2016, tifs2016}, whereas obfuscation techniques add noise to users' data samples~\cite{yoshida2019optimal,obf_1,obf_2,obf_3,randomizedresponse, takbiri2018matching, takbiri2017limits}.

Privacy can be compromised by linking the characteristics of a target sequence of activities to previously observed user behavior.  To provide privacy guarantees in the presence of such potential \emph{sequence matching}, a stochastic model for the users' data (e.g., Markov chains) has been generally assumed~\cite{mangold1999applying,juang1990segmental,farhad, negar_1, negar_2, Nazanin_WCNC2019}, and privacy attacks that match the statistical characteristics of the target sequence to those of past sequences of the user are considered.  These previous approaches have two limitations:  (i) many privacy attacks are based on simple ``pattern matching" for identification~\cite{keen}, classification~\cite{becker2011route, xu2011acculoc}, or prediction~\cite{eagle2009methodologies}, where the adversary (algorithm) looks (deterministically) for a specific ordered sequence of values in the user's data; and, (ii) as \textit{Privacy-Protection Mechanism} (PPM) designers, we may not know the underlying statistical model for users' data.  In particular, Takbiri et al.~\cite{takbiri2018matching} have shown that modeling errors can destroy privacy guarantees. 

We consider the following important question:  Can we thwart privacy attacks that de-anonymize users by finding specific identifying patterns in their data, \emph{even if we do not know what patterns the adversary might be exploiting}, and can we do so without assuming a certain model (or collection of models) for users' data?  Our privacy metric and the resulting obfuscation approach are based on the following idea: \emph{noise should be added in a way that the obfuscated user data sequences are likely to have a large number of common patterns}. This means that for any user and for any potential pattern that the adversary might obtain for that user, there will be a large number of other users with the same data pattern in their obfuscated sequences.  By focusing on this common type of privacy attack (pattern matching), the PPM is able to eliminate the need for making specific assumptions about the users' data model. 

%This in turn can be used to provide privacy guarantees against pattern matching attacks.  

%More precisely, we consider a scenario where $n$ length-$m$ sequences are being generated. We do not know anything about the way the sequences are constructed (e.g., the distribution of the data points in the sequences). All we know is that these $n$ sequences, labeled $1,2,\ldots,n$, each have $m$ elements drawn from the same alphabet-$\mathcal{R}$  (where $|\mathcal{R}|=r$). The sequences are revealed to a potential adversary but the labels are hidden. However, the adversary has obtained a ``pattern" (later defined precisely) in one of these sequences and would like to identify that sequence based on that pattern. For instance, in a fingerprinting attack, the webpages visited by users are used as the identifying patterns. The challenge is to design an obfuscation mechanism to apply to these sequences before they are revealed to the adversary to prevent such identification. The difficulty is that the PPM has very limited knowledge about these sequences, only knowing the values of $n$, $m$, and $\mathcal{R}$, and the PPM does not know which user the adversary might be seeking to identify or what pattern the adversary might employ. And, of course, the PPM wants the distortion to be as small as possible to preserve utility.  

To achieve privacy guarantees, we first introduce a data-independent obfuscation approach, which means that sequence obfuscation can be performed without knowledge of the actual data values in a user's sequence.  Such data independence would be of value in various applications where the upcoming events are not known a priori, for instance in website fingerprinting~\cite{shirani2018optimal} \cite{wondracek2010practical} and flow correlation~\cite{zhu2004flow} applications where obfuscations need to be applied on live (non-buffered) network packets.  Our approach relies on the concept of \textit{superstring}s, which contain every possible pattern of length less than or equal to the pattern length $l$ (repeated symbols are allowed in each contiguous substring). This in turn happens to be related to a rich area in combinatorics \cite{newey1973notes,radomirovic2012construction,johnston2013minimal,houston2014tackling}. 
%Of relevant interest are the De Bruijn sequences~\cite{deBruijn1946Acombinatiorials}, which give the answer for the shortest cyclic sequence (tail-to-head ligation allowed) that contains every possible pattern of length less than or equal to $l$, and which can be readily modified to yield the shortest \textit{superstring} when tail-to-head ligation is not allowed.  
After introducing and characterizing our data-independent obfuscation approach, we introduce data-dependent obfuscation approaches for comparison; data-dependent approaches are able to look at the values in the users' data sequences and base their obfuscation on such information.  Such data dependence would be possible in applications where the whole vector of user data is known to the obfuscation party at the obfuscation time, for instance image processing applications. 

The contributions of our paper are:
\begin{itemize}
\item We propose a formal framework for defending against pattern matching attacks when there is no statistical model for the user data (Section~\ref{sec:framework}).
\item We present a data-independent obfuscation approach based on \textit{superstring}s and, by lower bounding its performance for two different types of \textit{superstring}s, prove that it yields a non-zero fraction of user sequences that contain a potentially identifying pattern (Section~\ref{sec:simple-ppm}).
\item We develop data-dependent obfuscation approaches for our pattern matching framework (Section~\ref{sec:data-dependent}).
\item We validate the developed approaches on both synthetic data and the Reality Mining dataset to demonstrate their utility and compare their performance (Section~\ref{sec:numerical}).
\end{itemize}

%   DG:  Paragraph below will likely need to be modified later.

%After reviewing related work in detail in Section \ref{related_work}, we present the system model, definitions, and metrics employed in this paper in Section~\ref{sec:framework}.  Then, in Section~\ref{sec:simple-ppm}, we present our constructions and analyses for privacy guarantees for our data-independent obfuscation methods.  We introduce data-dependent approaches in Section \ref{sec:data-dependent}.  Section \ref{sec:numerical} provides a numerical evaluation of all of the approaches on synthetic data and on the {\em Reality Mining} dataset.  

Finally, we present the conclusions that can be drawn from our study in Section~\ref{conclusion}.

%% file: relatedWork.tex
\section{Related Work}
\label{related_work}
IoT devices provide important services but they have multiple potential privacy issues: 1) enabling unauthorized access and misuse of personal information~\cite{ukil2014iot},~\cite{pan2018flowcog}, 2) facilitating attacks on other systems~\cite{staff2015internet},~\cite{sadeghi2015security}, 3) creating personal privacy and safety problems~\cite{pan2015not},~\cite{lin2016iot},~\cite{zheng2018user}.  
%Mechanisms for protecting users' privacy can be classified into two main classes: 1) anonymization, and 2) obfuscation.

Anonymization conceals the mapping between users' identity and data by periodically changing the mapping to prevent statistical inference attacks. The $k$-anonymity protection approach is proposed in~\cite{samarati1998generalizing},~\cite{sweeney2002k}, which guarantees that the information for any person contained in the released version of the data cannot be distinguished from at least $k-1$ individuals. To overcome some drawbacks of $k$-anonymity, $l$-diversity is proposed in~\cite{xiao2006personalized}. Anonymizing social network data, which is more challenging than anonymizing relational data, is studied in~\cite{zhou2008preserving}.  The concept of perfect location privacy is introduced and characterized by Montazeri et al.~\cite{tifs2016}. 

%who demonstrate how anonymization should be used by location-based services (LBSs) systems to achieve the defined perfect location privacy.

%Though anonymization techniques protect users' privacy by modifying labels on users' data sequences, 

Anonymization is often insufficient because an adversary can track users' identities by using users' trajectory information~\cite{hoh2005protecting},~\cite{keen} or specific patterns~\cite{Naini2016},~\cite{karp1972rapid},~\cite{weiner1973linear}. Privacy preservation could be challenged by de-anonymization attacks if some sensitive information is known by an adversary~\cite{unnikrishnan2014asymptotically},~\cite{zhou2008brief},~\cite{unnikrishnan2013anonymizing}.  Hence, obfuscation techniques protect users' privacy by introducing perturbations into users' data sequences to decrease their accuracy~\cite{ardagna2009obfuscation}.  In \cite{obf_1}, the authors propose an adaptive algorithm which adjusts the resolution of location information along spatial and temporal dimensions. In \cite{obf_2}, a comprehensive solution aimed at preserving location privacy of individuals through artificial perturbations of location information is presented.  The work of \cite{obf_3} provides efficient distributed protocols for generating random noise to provide security against malicious participants. In \cite{randomizedresponse}, a randomized response method is proposed which allows interviewees to maintain privacy while increasing cooperation. 

Data protection mechanisms might limit the utility when data also needs to be shared with an application or the provider to achieve some utility~\cite{wang2017estimation} or a quality of service constraint~\cite{hoh2005protecting}. 
%For instance, when a user has uploaded her ratings for a restaurant, she will receive new recommendations in the future for the restaurants around her which fit her preferences (utility).  
Theoretical analyses of the privacy-utility trade-off (PUT) are provided in~\cite{diaz2019robustness},~\cite{yoshida2019optimal}. A key concept of \textit{relevance} is proposed which strikes a balance between the need of service providers, requiring a certain level of location accuracy, and the need of users, asking to minimize the disclosure of personal location information~\cite{obf_2}. In~\cite{takbiri2018matching},~\cite{takbiri2017limits}, Takbiri et al. derive the theoretical bounds on the privacy versus utility of users when an adversary is trying to perform statistical analyses on time series to match the series to user identity. 

Pattern matching problems have attracted researchers in recent years, in particular fast pattern matching~\cite{karp1972rapid},~\cite{weiner1973linear}, database search~\cite{faloutsos1994fast},~\cite{ma2011network},~\cite{al2002structural}, secure pattern matching~\cite{yasuda2013secure},~\cite{wang2017privacy},~\cite{baron20125pm}, and identification~\cite{guan2020sequence},~\cite{keen},~\cite{becker2011route},~\cite{rahmat2017file} and characterization~\cite{povinelli2003new, gonzalez2008understanding, isaacman2010tale, keralapura2010profiling} by using patterns or subsequences. 

%% file: model.tex
\section{System Model, Definitions, and Metrics}
\label{sec:framework}
%Here, we employ a framework similar to~\cite{tifs2016,takbiri2018matching,nazanin_ISIT2018, ISIT18-longversion}. 

Consider a system with $n$ users whose identification we seek to protect.  Let $X_u(k)$ denote the data of user $u$ at time $k$.  We assume there are $r \geq 2$ possible values for each of the users' data points in a finite size set $\mathcal{R}=\{0,1, \ldots, r-1\}$.
Let $\textbf{X}_u$ be the $m \times 1$ vector containing the data points of user $u$, and $\textbf{X}$ be the $m \times n$ matrix with the $u^{th}$ column equal to $\textbf{X}_u:$
\begin{align}
\no \textbf{X}_u = [X_u(1), X_u(2), \cdots, X_u(m)]^T,\ \ \ \textbf{X} =\left[\textbf{X}_{1}, \textbf{X}_{2}, \cdots, \textbf{X}_{n}\right].
\end{align}
As shown in Fig.~\ref{fig:xyz}, in order to achieve privacy for users, both anonymization and obfuscation techniques are employed. In Fig.~\ref{fig:xyz}, $\textbf{Z}$ denotes the reported data of the users after applying the obfuscation, and $\textbf{Y}$ denotes the reported data after applying the obfuscation and the anonymization, where, with $Z_u(k)$ denoting the obfuscated data of user $u$ at time $k$ and $Y_u(k)$ denoting the obfuscated and anonymized data of user $u$ at time $k$, respectively, $\textbf{Z}$ and $\textbf{Y}$ are defined analogously to $\textbf{X}$.

\begin{figure}[h]
	\centering
	\includegraphics[width = 1\linewidth]{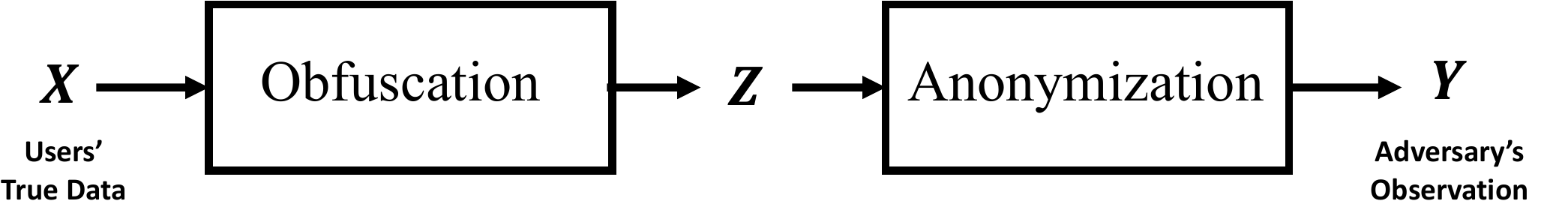}
	\caption{Applying obfuscation and anonymization techniques to the users' data points.}
	\label{fig:xyz}
\end{figure}

Next, we provide a formal definition of a \textit{pattern}.  As an example, a potential pattern could be the sequence of locations that the user normally visits in a particular order: their office, the gym, a child's school.  The visited locations might not necessarily be contiguous in the sequence, but they are close to each other in time. Hence, we impose two conditions on a \textit{pattern}: first, the elements of the pattern sequence must be present in order. Second, consecutive elements of the pattern sequence must appear within \textit{distance} less than or equal to $h$, where the \textit{distance} between two elements is defined as the difference between the indices of those elements ($h \geq 1$).  The parameter $h$ could have value one for the most restricted case: the elements of the pattern sequence must appear consecutively in users' sequences.  And $h$ could be infinity for the unconstrained case: applications which do not consider \textit{distance} for detecting a pattern, e.g., traffic analysis.

\begin{define}
A \textit{pattern} is a sequence $\textbf{Q}=q^{(1)} q^{(2)} \cdots{}q^{(l)}$, where $q^{(i)} \in \{0,1, \cdots, r-1\}$ for all $i\in{}\{1,2,\cdots,l\}$.  A user $u$ is said to have the \textit{pattern} $\textbf{Q}$ if
\begin{itemize}
 \item The sequence $\textbf{Q}$ is a subsequence (not necessarily of consecutive elements) of user $u$'s sequence. 
 \item For each $i \in \{1,2,\cdots, l-1\}$, $q^{(i)}$ and $q^{(i+1)}$ appear in user $u$'s sequence with \textit{distance} less than or equal to $h$.
\end{itemize}
\label{def:pattern}
\end{define}

\hspace{-11.5 pt}\textbf{Obfuscation Mechanism:}
Given the model above and the definition of {$\epsilon-$}\textit{privacy} below, the objective is to design obfuscation schemes for the user data sequences that maximize {$\epsilon-$}\textit{privacy} with minimum sequence distortion without knowing what pattern the adversary might be exploiting or which user the adversary might be targeting.  For simplicity, we consider the case of sparsely sampled data and thus leave to future work the enforcement of consistency constraints on the obfuscated user data sequences, such as a continuity constraint requiring that adjacent sequence elements have similar values.  The design and characterization of obfuscation mechanisms is the main topic of the succeeding sections.

\hspace{-11.5 pt}\textbf{Anonymization Mechanism:}
Anonymization is modeled by a random permutation $\Pi$ on the set of $n$ users, $\mathcal{U} = \{1,2,\cdots,n\}$. Each user $u$ is anonymized by the pseudonym function $\Pi(u)$.  Per above, $\textbf{Y}$ is the anonymized version of $\boldsymbol{Z}$; thus,
\begin{IEEEeqnarray*}{rCl}
\textbf{Y} && = \text{Perm}(\textbf{Z}_1, \textbf{Z}_2, \cdots, \textbf{Z}_n;\Pi)\\
&& = [\textbf{Z}_{\Pi^{-1}(1)}, \textbf{Z}_{\Pi^{-1}(2)}, \cdots, \textbf{Z}_{\Pi^{-1}(n)}]\\
&& = [\textbf{Y}_1, \textbf{Y}_2, \cdots, \textbf{Y}_n],
\end{IEEEeqnarray*}
where $\text{Perm}(~\cdot{}~; \Pi)$ is the permutation operation with permutation function $\Pi$. As a result, $\textbf{Y}_u = \textbf{Z}_{\Pi^{-1}(u)}$ and $\textbf{Y}_{\Pi(u)} = \textbf{Z}_u$.  In practice, our model would arise when anonymization takes place every $m$ samples.  In such a case, it is sufficient to assume sequences of length $m$ and employ the anonymization once to conceal the mapping between users and their data sequences.

\hspace{-11.5 pt}\textbf{Adversary Model:} The adversary has access to a sequence of observations of length$-m$ for each user; in other words, for each $u \in \{1,2, \cdots,n\}$, the adversary observes $Y_{\Pi(u)}(1),Y_{\Pi(u)}(2), \cdots, Y_{\Pi(u)}(m)$. We also assume the adversary has identified a \textit{pattern} $\textbf{Q}_v$ of a specific user $v$, $q^{(1)}_vq^{(2)}_v\cdots q^{(l)}_v$, and is trying to identify the sequence of a user $v$ by finding the sequence with pattern $q^{(1)}_vq^{(2)}_v\cdots q^{(l)}_v$.  The adversary knows the obfuscation and the anonymization mechanisms; however, they do not know the realization of the random permutation $\left(\Pi\right)$ and they do not know the realization of any randomly generated elements of the obfuscation mechanism. 
%such as the generated \textit{superstring} $\left(\textbf{a}_u, \ u \in \mathcal{U}\right)$ for the constructions of Section \ref{sec:simple-ppm}.  An adversary identifies a \textit{pattern} by Definition~\ref{def:pattern}.

We define {$\epsilon-$}\textit{privacy} as:
\begin{define}
	User $v$ with data \textit{pattern} $q^{(1)}_vq^{(2)}_v\cdots q^{(l)}_v$ has {$\epsilon-$}\textit{privacy} if for any other user $u$, the probability that user $u$ has \textit{pattern} $q^{(1)}_vq^{(2)}_v\cdots q^{(l)}_v$ in their obfuscated data sequence is at least $\epsilon$.
\end{define}
Loosely speaking, this implies that the adversary cannot identify user $v$ with probability better than $\frac{1}{n\epsilon}$.  If we assume $\epsilon$ is a constant independent of $n$, $\epsilon-$privacy is a strong requirement for privacy - equivalent to {$n\epsilon-$}anonymity in the setting of $k$-anonymity.  In contrast, in perfect privacy \cite{tifs2016, takbiri2017limits, takbiri2018matching} it suffices that each user is confused with $N^{(n)}$ users, where $N^{(n)} \rightarrow \infty$ as $n \rightarrow \infty$.  Hence, we will also consider cases where $\epsilon$ is a decreasing function of $n$ so as to consider less stringent privacy definitions.

%% file: data_independent.tex
\section{Privacy Guarantee for Model-Free PPMs} \label{sec:simple-ppm}

We present constructions for model-free privacy-protection mechanisms under the model of Section \ref{sec:framework} and then characterize their performance. 

%Two \textit{data-independent obfuscation} (DIO) algorithms: Superstring-Based Obfuscation (SBU), and Independent and Identically Distributed (i.i.d.) Obfuscation, are presented.

\subsection{Constructions}

%Since we do not know which pattern the adversary might be using to identify a user, we want the obfuscated sequence of each user to include a large number of patterns.  

For any user and for any potential pattern that the adversary might obtain for that user, we want to ensure there will be a large number of other users with the same data pattern in their obfuscated data sequences.  First we define the concept of a \textit{superstring} and then our obfuscation mechanism.

\begin{define}
A sequence is an $(r,l)-$\textit{superstring} if it contains all possible $r^l$ length-$l$ strings (repeated symbols allowed) on a size-$r$ alphabet-$\mathcal{R}$ as its contiguous substrings (cyclic tail-to-head ligation not allowed). \label{def_superstring}
\end{define}
We define $f(r,l)$ as the length of the shortest $(r,l)-$\textit{superstring}. A trivial upper bound is $f(r,l) \leq lr^l$, as $lr^l$ is the length of the $(r,l)-$\textit{superstring} obtained by concatenating all possible $r^l$ substrings.  As an example of a \textit{superstring}, the sequence $11221$ is a $(2,2)-$\textit{superstring} because it contains $11$, $12$, $21$, and $22$ as its contiguous subsequences; thus $f(2,2) \leq 5 \leq 8 = lr^l$. 

%We now introduce our \emph{superstring-based obfuscation mechanism}.  

\hspace{-11.5 pt}\textbf{Superstring-Based Obfuscation (SBU):}
Recall that $\textbf{Z}_u$ is the $m \times 1$ vector of the obfuscated version of user $u$'s data sequence, and $\textbf{Z}$ is the $m \times n$ matrix with $u^{th}$ column $\textbf{Z}_u$:
\begin{align}
\no \textbf{Z}_u = [Z_u(1), Z_u(2), \cdots, Z_u(m)]^T,\ \ \ \textbf{Z} =\left[\textbf{Z}_{1},  \textbf{Z}_{2}, \cdots, \textbf{Z}_{n}\right].
\end{align}

The basic procedure is shown in Fig.~\ref{Obfus_General}.  For each user, we independently and randomly generate an $(r,l)-$\textit{superstring} from the \textit{superstring} solution set described below.  We denote the generated $(r,l)-$\textit{superstring} as $\boldsymbol{a}_u = \{a_u(1),a_u(2), \cdots, a_u({L_s}) \}$, where $L_s $ is the length of the generated \textit{superstring}.  The parameter $p_{\text{obf}}$ is the probability that we will change a given data sample.  Thus, for each data point of each user, we independently generate a Bernoulli random variable ${W}_u(k)$ with parameter $p_{\text{obf}}$. As shown in Fig.~\ref{Obfus_General}, the obfuscated version of the data sample of user $u$ at time $k$ can then be written as:
\begin{align}
\no Z_u(k) = \left\{
\begin{array}{rl}
X_u(k), & \text{if }W_u(k)=0 \\ \ a_u(j), & \text{if }W_u(k)=1,
 \end{array} \right.
\end{align}
	where $j = \sum\limits_{k'=1}^{k}W_u(k')$, and $a_u(j)$ is the $j^{th}$ element of the $(r,l)-$\textit{superstring} used for the obfuscation.
If the length of the generated $(r,l)-$\textit{superstring} is not sufficient (i.e. $\sum\limits_{k'=1}^m W_u(k') > L_s$), we choose another \textit{superstring} at random to continue. 

\begin{figure}[htbp]
	\centerline{\includegraphics[width=1\linewidth]{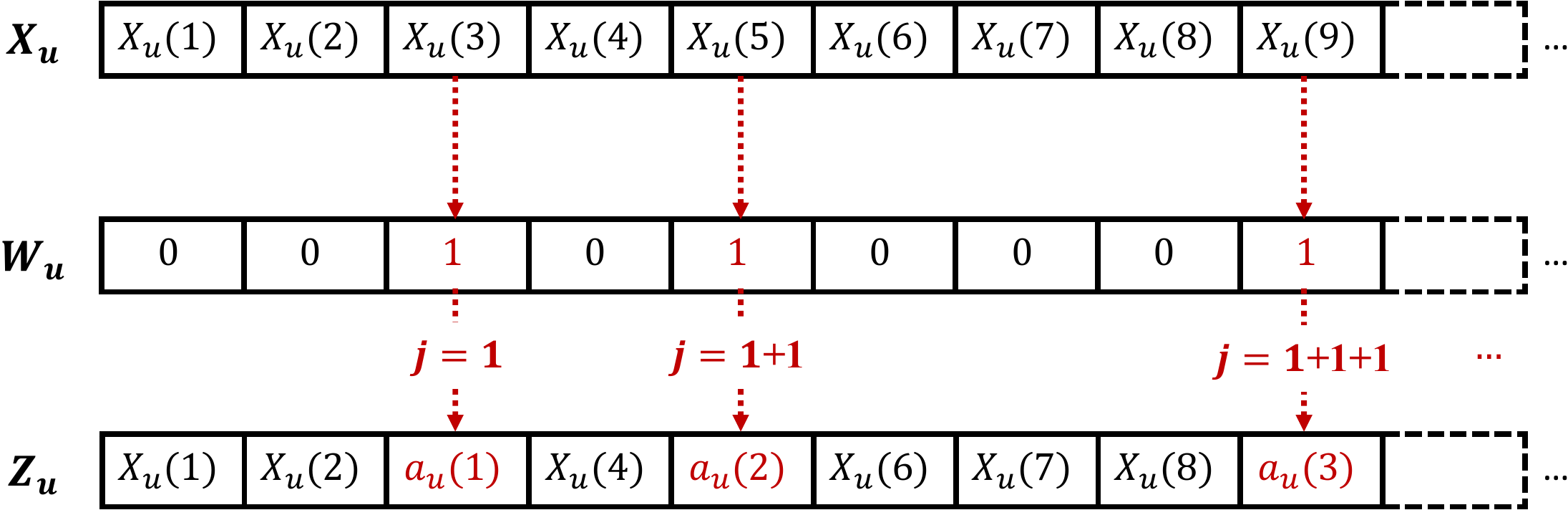}}
	\caption{The obfuscation of the data sequence of user $u \in \mathcal{U}$ based on an $(r,l)-$\textit{superstring}.}
	\label{Obfus_General}
\end{figure}

%%  HERE

\hspace{-11.5 pt}\textbf{Independent and Identically Distributed (i.i.d.) Obfuscation:}

In ~\cite{takbiri2018matching,nazanin_ISIT2018, ISIT18-longversion} a \emph{uniform i.i.d.} obfuscation mechanism is used. For each user, an i.i.d. sequence of random variables $\mathbf{b_u}= \{b_u(1),b_u(2), \cdots \}$ uniformly distributed on the alphabet $\mathcal{R}=\{0,1, \ldots, r-1\}$ is generated.  These values are used to obfuscate the sequence $X_u(k)$: for each data point of each user, we independently generate a Bernoulli random variable ${W}_u(k)$ with parameter $p_{\text{obf}}$. The obfuscated version of the data sample of user $u$ at time $k$ can be written as:
\begin{align}
\no Z_u(k) = \left\{
\begin{array}{rl}
X_u(k), & \text{if }W_u(k)=0 \\ \ b_u(j), & \text{if }W_u(k)=1,
 \end{array} \right.
\end{align}
where $j = \sum\limits_{k'=1}^{k}W_u(k')$.  The i.i.d. obfuscation will be a benchmark for comparison of our superstring-based approaches.

%Let $\mathbb{P}_{\text{i.i.d.}}$ be the probability that the obfuscated sequence has user $1$'s identifying pattern after applying the i.i.d. obfuscation. When $1=1$, we have: \[ \mathbb{P}_{\text{i.i.d.}}(l=1) \geq \sum_{k=0}^{m} {m \choose k} p_{\text{obf}}^k(1-p_{\text{obf}})^{m-k} \left(1-(1-\frac{1}{r})^k\right),\] where $m$ is the length of the sequence.

\subsection{Analysis}

Without loss of generality, consider $\epsilon-$privacy for user $1$ with pattern sequence $q^{(1)}_1q^{(2)}_1\cdots{}q^{(l)}_1$. 
%Per Section~\ref{sec:framework}, we want to study whether the obfuscated data sequences of other users are likely to have the same data pattern as user $1$ to confuse a pattern-matching adversary trying to find user $1$. 
%Let $\mathcal{B}_u$ be the event that user $1$'s pattern $q_1^{(1)} q^{(2)}_1 \cdots q_1^{(l)}$ appears in user $u$'s obfuscated data points $\boldsymbol{Z}_u$ due to our obfuscation technique.  
The pattern length $l$ and the maximum \textit{distance} $h$ between the appearance of pattern elements are assumed to be known and treated as constants, but we hasten to note that this defends against an attacker employing a pattern with length less than or equal to $l$ and maximum \textit{distance} greater than or equal to $h$.  

We assume a worst-case scenario: user 1 has a pattern unique to their data set that can be exploited for identification. %in other words, no other user has the identifying pattern in their data sequence before obfuscation.  
We start with the upper bound $lr^l$ for the length of an $(r,l)-$\textit{superstring}.  We will prove that such a \textit{superstring} guarantees that at least a certain fraction $\epsilon$ of users will have the same pattern as user $1$ after employing the obfuscation mechanism.  Later, we will improve this result by introducing the De Bruijn sequence to shorten the \textit{superstring}. 

\begin{define}
Let $\mathcal{B}_u$ be the event that the obfuscated sequence $\textbf{Z}_u$ has user $1$'s identifying pattern due to obfuscation by an $(r,l)-$\textit{superstring} with length $lr^l$ obtained by concatenating all possible $r^l$ substrings.
\end{define}

\begin{thm} \label{thm:simple-ppm-privacy-guarantee}
The probability of $\mathcal{B}_u$, denoted by $\mathbb{P}\left(\mathcal{B}_u\right)$, is lower bounded by a constant that does not depend on $n$ as:
 \begin{align}
\mathbb{P}\left(\mathcal{B}_u\right) \geq
\frac{\left( 1-\left( 1-p_{\text{obf}}\right)^h \right)^{(l-1)}}{r^l} \sum_{\alpha=0}^{\min \left\{(r^l-1), \big{\lfloor} \frac{Gp_{\text{obf}}}{l}\big{\rfloor} \right\} } \hspace{-25 pt}1-\exp \left(-\frac{ \delta_\alpha^2}{2} Gp_{\text{obf}} \right), \label{neq_th1}
\end{align}
where
\[G=m-h(l-1), \ \ \ \delta_\alpha=1-\frac{\alpha l}{Gp_{\text{obf}}},~~ \text{ for } \alpha=0,1,\cdots, r^l-1. \]
\end{thm}
\begin{proof}
The notation employed here within the procedure of obfuscation for user $u\in{}\mathcal{U}$ is shown in Fig.~\ref{obfu-fig}.  Note that our generated \textit{superstring} can have more than one copy of each pattern, but we pessimistically focus on one copy of our desired pattern.  We denote $L_{u,1}$ as the index of the first element of the pattern of the \textit{superstring} of user $u$, such that $a_u(L_{u,1}) = q^{(1)}_1, a_u(L_{u,1}+1) = q^{(2)}_1,\cdots,a_u(L_{u,1}+l-1) = q^{(l)}_1$, and correspondingly, $M^i_{u,1}$ is the index of the data point $X_u(M^{(i)}_{u,1})$ that is obfuscated to $q^{(i)}_1$ ($M^{(i)}_{u,1} < m$), for $i=1,2,\ldots,l$:
 \begin{equation}
 Z_u(M_{u,1}^{(i)}) = a_u(L_{u,1} + i-1) = q_1^{(i)}, \text{ for any } u \in \mathcal{U}.
 \end{equation}
The sequences $X_u$ and $Z_u$ can be assumed to be infinitely long with the adversary only seeing the first $m$ elements of $Z_u$. Therefore, a sufficient condition for $\mathcal{B}_u$ (according to Definition~\ref{def:pattern}) is $\mathcal{E}_u \bigcap \mathcal{F}_u$, where:

%According to Definition~\ref{def:pattern}, the event $\mathcal{B}_u$ occurs if (but not only if) the following two events occur: (i) the user $1$'s pattern $q_1^{(1)} q^{(2)}_1 \cdots q_1^{(l)}$ appears in user $u$'s obfuscated data points $\boldsymbol{Z}_u$; and, (ii) the \textit{distance} between any neighboring points of pattern $q_1^{(1)} q^{(2)}_1 \cdots q^{(l)}_1$ in $\boldsymbol{Z}_u$ is smaller than or equal to $h$. Now, if we accordingly define event $\mathcal{E}_u$ and $\mathcal{F}_u$ as:
\begin{align}
\mathcal{E}_u: M^{(1)}_{u,1} \leq m - h(l-1) = G,\ \
\label{e_n4}
\end{align}
\begin{align}
\mathcal{F}_u: D^{(1)}_u \leq h; D^{(2)}_u \leq h;\cdots,D^{(l-1)}_u\leq h,\ \
\label{f_n4}
\end{align}
where $D^{(i)}_u = M^{(i+1)}_{u,1} - M^{(i)}_{u,1}$ are the \textit{distance}s between $q^{(i+1)}_1$ and $q_1^{(i)}$ in user $u$'s obfuscated sequence $\boldsymbol{Z}_u$, for $i=1,2,\ldots,l-1$.
Note that we have defined $\mathcal{E}_u$ and $\mathcal{F}_u$ so as to make them independent.  Thus, we have
\begin{equation}
\mathbb{P}\left(\mathcal{B}_u\right) \geq \mathbb{P}\left(\mathcal{E}_u\right) \mathbb{P}\left(\mathcal{F}_u\right). \label{greater}
\end{equation}

%thus,

%\begin{equation*}

%\mathbb{P}\left(\mathcal{E}_u\right) = \mathbb{P}\left(\text{at least }L_{u,1}\text{ success in }G\text{ trials}\right).

%\end{equation*}

The probability of event $\mathcal{E}_u$ is the probability of $L_{u,1}$ successes in $M$ Bernoulli trials, where each trial has probability of success $p_{\text{obf}}$. 
Since each user employs a randomly chosen \textit{superstring} for obfuscation, the pattern is equally likely to be in any of the $r^l$ substrings of length $l$; hence,
\begin{align}
\mathbb{P}\left(L_{u,1} = \alpha l +1\right) &= \frac{1}{r^l}, \quad \alpha= 0,1,\cdots,r^l-1. \ \
\label{e_n}
\end{align}
Thus, by employing the Law of Total Probability, we have:
\begin{align}
\no \mathbb{P}\left(\mathcal{E}_u\right) &= \sum_{\alpha=0}^{r^l-1}\mathbb{P}\left(\text{at least }L_{u,1}\text{ success in }G\text{ trials} \Big{|} L_{u,1} =\alpha l+1\right) \\
\no &\hspace{24 pt}\cdot  \mathbb{P}\left(L_{u,1}=\alpha l+1\right) \\
 \no &= \frac{1}{r^l}\sum_{\alpha=0}^{r^l-1}\mathbb{P}\left(\text{at least }\alpha l+1\text{ success in }G\text{ trials}\right) \\
 \no &= \frac{1}{r^l}\sum_{\alpha=0}^{r^l-1}\left[1- \mathbb{P}\left(\text{less than }\alpha l+1\text{ success in }G\text{ trials}\right)\right].
\end{align}
Define $\mathcal{A}_\alpha$ as the event that there exists less than $\alpha l+1$ successes in $G$ trials. By employing the Chernoff Bound:
\begin{align}
p(\mathcal{A}_\alpha) \leq \exp\left(-\frac{1}{2}\delta_\alpha^2Gp_{\text{obf}}\right), \quad \text{for all } \alpha < \frac{Gp_{\text{obf}}}{l}.\ \
\label{chern}
\end{align}
Now, by using (\ref{e_n}) and (\ref{chern}):
\begin{equation}
\mathbb{P}\left(\mathcal{E}_u\right) \geq \frac{1}{r^l}\sum_{\alpha=0}^{\min\left\{(r^l-1), \big{\lfloor}{}\frac{Gp_{\text{obf}}}{l}\big{\rfloor}
\right\}}1 - \exp\left(-\frac{1}{2}\delta_\alpha^2Gp_{\text{obf}}\right).\ \label{p_E_u}
\end{equation}
Note that sub-events of $\mathcal{F}_u$: $D_u^{(1)} \leq h,\cdots,D_u^{(l-1)}\leq h$ are independent; thus, the probability of event $\mathcal{F}_u$ is:
\begin{align}
\mathbb{P}\left(\mathcal{F}_u\right) = \prod\limits_{i=1}^{l-1} \mathbb{P}\left(D_u^{(i)} \leq h\right) =\left(1 - (1-p_{\text{obf}})^h\right)^{(l-1)}. \label{p_F_u}
\end{align}
Thus, by \eqref{greater}, \eqref{p_E_u} and \eqref{p_F_u}, we obtain \eqref{neq_th1}.
\end{proof}

\begin{figure}[htbp]
	\centerline{\includegraphics[width=1.0\linewidth]{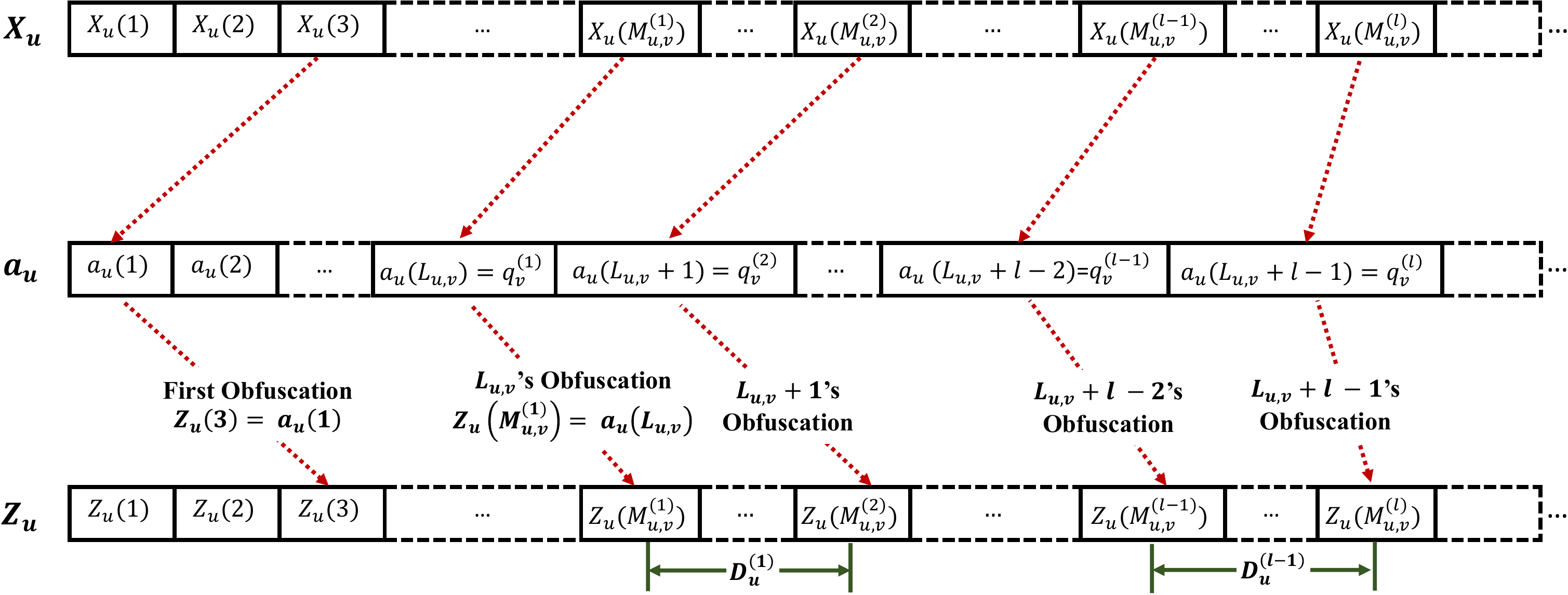}}
	\caption{The notation for the proof of Theorem 1 in the obfuscation of the trace of user $u \in \mathcal{U}$.}
	\label{obfu-fig}
\end{figure}

The methodology of Theorem~\ref{thm:simple-ppm-privacy-guarantee} can be applied with $(r,l)-$\textit{superstring}s of shorter length for stronger privacy guarantees. The following lemma provides a construction for the shortest $(r,l)-$\textit{superstring} and evaluates its length.

\begin{lem}
The length of the shortest $(r,l)-$\textit{superstring} is equal to $r^l + l - 1$; that is, $f(r,l) = r^l + l - 1$.
\label{lem1}
\end{lem}
\begin{figure}[htbp]
	\centerline{\includegraphics[width=0.95\linewidth]{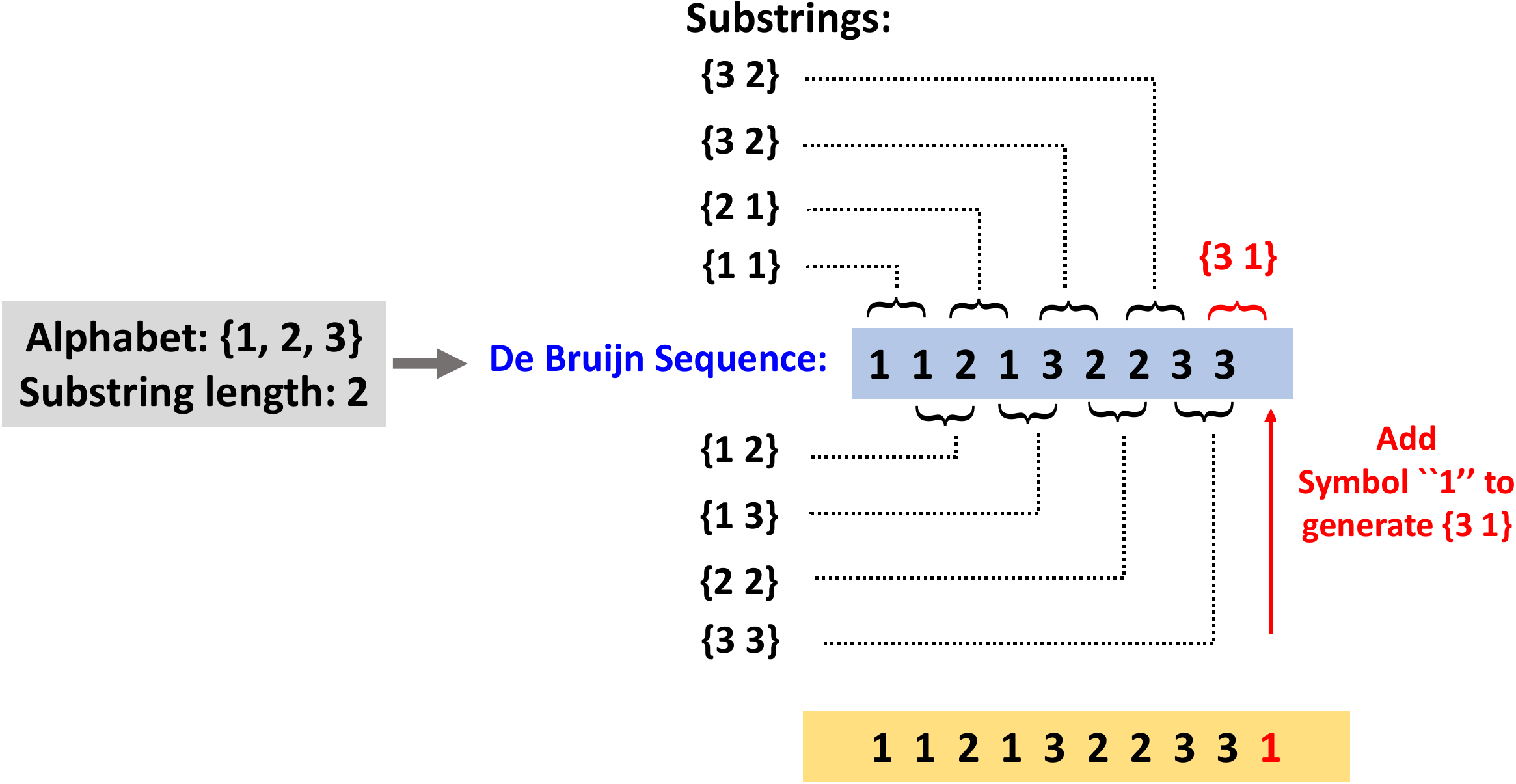}}
	\caption{The construction of a shortest $(3,2)-$\textit{superstring} by using a De Bruijn sequence $B(3,2)$. The length of the constructed $(3,2)-$\textit{superstring} is $f(3,2) = 3^2 + 2-1 = 10$.}
	\label{fig:string2}
\end{figure}
\begin{proof}

%A (non-unique) shortest $(r,l)-$\textit{superstring} can be constructed by a De Bruijn sequence , which we denote $B(r,l)$. 
We denote by $B(r,l)$ a De Bruijn sequence~\cite{deBruijn1946Acombinatiorials, annexstein1997generating} of order $l$ on a size-$r$ alphabet-$\mathcal{R}$.  A De Bruijn sequence is a sequence with length $r^l$ in which every possible length-$l$ substring on $\mathcal{R}$ occurs exactly once as a contiguous subsequence, given that the last $(l-1)$ and the first $(l-1)$ letters of the De Bruijn sequence form a cyclic tail-to-head ligation for counting the substrings.  

We construct a shortest $(r,l)-$\textit{superstring} with length $(r^l + l - 1)$ from a chosen De Bruijn sequence $B(r,l)$ by repeating $B(r,l)$'s front $(l-1)$ symbols at the end of the sequence.  We first prove that the constructed sequence is an $(r,l)-$\textit{superstring}. The sequence has the first $[r^l - (l-1)]$ substrings because it contains a full De Bruijn sequence $B(r,l)$ in its first $r^l$ symbols. In addition, since the left $(l-1)$ substrings in $B(r,l)$ are counted by tracking from the last $(l-1)$ letters and the first $(l-1)$ letters as mentioned, the left $(l-1)$ substrings also appear in the constructed \textit{superstring} in a non-cyclic way, since the De Bruijn sequence's front $(l-1)$ symbols have been copied to its end. Thus, the constructed sequence contains all possible $r^l$ substrings, and hence, by Definition~\ref{def_superstring}, it is a valid $(r,l)-$\textit{superstring}.

Next we prove that the constructed sequence gives the shortest solution for an $(r,l)-$\textit{superstring}. Each of the distinct substrings on the size-$r$ alphabet-$\mathcal{R}$ must start at a different position in the sequence, because substrings starting at the same position are not distinct. Therefore, an $(r,l)-$\textit{superstring} must have at least $(r^l+l-1)$ symbols. 
%Since the constructed $(r,l)-$\textit{superstring} has exactly $(r^l+l-1)$ symbols, it is optimally short with length $(r^l+l-1)$.
\end{proof}
The solution for the shortest $(r,l)-$\textit{superstring} is non-unique in general for $r\geq{}2$ since we can construct our $(r,l)-$\textit{superstring} by taking any De Bruijn sequence $B(r,l)$. 
%(which is also non-unique: another De Bruijn sequence can be generated by circular shifting $B(r,l)$ in the left or right direction by any number of digits) from any De Bruijn sequence pattern set. An example of construction of a shortest $(3,2)-$\textit{superstring} by a De Bruijn sequence $B(3,2) = \text{``}112132233\text{''}$ is shown in Fig.~\ref{fig:string2}, and its is length $f(3,2)=10$. Another solution can be generated, for instance, by right circular shifting the De Bruijn sequence $B(3,2)$ by one digit and copy the first symbol `$3$' to its end: ``$3112132233$''. Another solution can be: ``$3311213223$'' (by further right circular shifting $B(3,2)$ by one digit and adding its front symbol `$3$' to its end). Next we consider the privacy performance when the shortest $(r,l)-$\textit{superstring} is employed.

\hspace{-11.5 pt}\textbf{Shortest-length Superstring-Based Obfuscation (SL-SBU):}
For each user we randomly (uniformly) choose a shortest-length \textit{superstring} (as described above) and employ it for obfuscation.  As noted earlier, if we reach the end of a \textit{superstring}, another one is chosen uniformly at random.

\begin{define}
Let $\mathcal{B}_u'$ be the event that the obfuscated sequence $\textbf{Z}_u$ has user $1$'s identifying pattern due to obfuscation by the shortest $(r,l)-$\textit{superstring} with length $f(r,l)$.
\end{define}

\begin{thm} \label{second-thm}
The privacy performance when the shortest $(r,l)-$\textit{superstring} is employed is given by:
\begin{align}
\mathbb{P}\left(\mathcal{B}_u'\right) \geq 
\frac{\left( 1-\left( 1-p_{\text{obf}}\right)^h \right)^{(l-1)}}{r^l} \sum_{\alpha=0}^{\min \left\{(r^l-1), \big{\lfloor} Gp_{\text{obf}}\big{\rfloor} \right\} }\hspace{-30 pt} 1-\exp \left(-\frac{\delta_\alpha'^2}{2} Gp_{\text{obf}} \right),
\label{eq:thm2}
\end{align}
where
 \[G=m-h(l-1), \ \ \delta_\alpha'=1-\frac{\alpha}{Gp_{\text{obf}}}, \text{ for } \alpha=0,1,\cdots, r^l-1. \]
\end{thm}
\begin{proof}
By using (\ref{greater}), we have:
\begin{equation}
\mathbb{P}\left(\mathcal{B}'_u\right) \geq \mathbb{P}\left(\mathcal{E}'_u\right) \mathbb{P}\left(\mathcal{F}'_u\right), \label{greater2}
\end{equation}
where the events $\mathcal{E}'_u$ and $\mathcal{F}'_u$ are defined analogously to the events $\mathcal{E}_u$ and $\mathcal{F}_u$ defined in (\ref{e_n4}) and (\ref{f_n4}), respectively.\\
For a given \textit{superstring} set generated by a De Bruijn sequence $B(r,l)$, we note that the index values $L_{u,1}$ are equally likely over the first $r^l$ indices in the $(r,l)-$\textit{superstring} chosen by user $u$, since one $(r,l)-$\textit{superstring} can be selected by uniformly circular shifting $B(r,l)$ by Lemma~\ref{lem1}. So we have:
\begin{align}
\mathbb{P}\left(L_{u,1} =\alpha+1\right) &= \frac{1}{r^l}, \quad \alpha = 0,1,\cdots,r^l-1. \ \
\label{e_n'}
\end{align}
\iffalse
Now, by employing the Law of Total Probability, we have:
\begin{align}
 \no \mathbb{P}\left(\mathcal{E}_u'\right) &= \sum_{\alpha=0}^{r^l-1}\mathbb{P}\left(\text{at least }L_{u,1}\text{ success in }G\text{ trials} \Big{|} L_{u,1} =\alpha+1\right) \\
\no &\hspace{24 pt}\cdot  \mathbb{P}\left(L_{u,1} =\alpha+1\right) \\
 \no &= \frac{1}{r^l}\sum_{\alpha=0}^{r^l-1}\mathbb{P}\left(\text{at least }\alpha+1\text{ success in }G\text{ trials}\right) \\
 \no &= \frac{1}{r^l}\sum_{\alpha=0}^{r^l-1}\left[1- \mathbb{P}\left(\text{less than }\alpha+1\text{ success in }M\text{ trials}\right)\right].
\end{align}
Define $\mathcal{A}_\alpha'$ as the event that there exists less than $\alpha+1$ success in $M$ trials. Employing a Chernoff Bound yields:
\begin{align}
p(\mathcal{A}_\alpha') \leq \exp\left(-\frac{1}{2}\delta_\alpha'^2Mp_{\text{obf}}\right), \quad \text{for all } \alpha < Mp_{\text{obf}}
\label{chern'}
\end{align}
\fi
Similarly, by employing a Chernoff Bound and the Law of Total Probability, we have:
\begin{align}
\mathbb{P}\left(\mathcal{E}_u'\right) \geq \frac{1}{r^l}\sum_{\alpha=0}^{\min\left\{(r^l-1), \big{\lfloor}{}Gp_{\text{obf}} \big{\rfloor}
\right\}}1 - \exp\left(-\frac{1}{2}\delta_\alpha'^2Gp_{\text{obf}}\right). \label{p_E_u'}
\end{align}
In addition, similar to~\eqref{p_F_u}, $\mathbb{P}\left(\mathcal{F}_u'\right) = \mathbb{P}\left(\mathcal{F}_u\right)$, which, combined with (\ref{greater2}) and (\ref{p_E_u'}), leads to \eqref{eq:thm2}.
\end{proof}

\begin{lem}
The lower bounds achieved by Theorem~\ref{thm:simple-ppm-privacy-guarantee} and Theorem~\ref{second-thm} are independent of the data sequence $\textbf{X}$ if the data point set $\mathcal{R}$ is known.
\end{lem}
\begin{proof}
This follows immediately from Theorems~\ref{thm:simple-ppm-privacy-guarantee} and \ref{second-thm}.
\end{proof}

Theorems~\ref{thm:simple-ppm-privacy-guarantee} and \ref{second-thm} provide $\epsilon$-privacy for constant $\epsilon$ (i.e. $\epsilon$ not decreasing in the number of users $n$).  As noted in Section \ref{sec:framework}, this is a very strong version of privacy, and hence weaker forms are also of practical interest.  Thus, we consider cases where $\epsilon$ goes to zero, but in a way that each user is still confused with $N^{(n)}$ users, where $N^{(n)} \rightarrow \infty$ as $n \rightarrow \infty$. First, the following lemma readily establishes that there are infinitely many users with the same pattern as user $1$ in such cases.

\begin{lem} \label{thm3}
Let $N^{(n)}$ be the number of users with the same pattern as user $1$. For any $0<\beta <1$, if $\mathbb{P}\left(\mathcal{B}_u'\right) = \mathbb{P}(\text{user } u\text{ has pattern of user } 1)   \geq \frac{1}{n^{1-\beta}}$, then $N^{(n)}\to \infty$ with high probability as $n \to \infty$.  More specifically, as $n \to \infty$
\[ \mathbb{P}\left(N^{(n)} \geq \frac{n^\beta}{2} \right) \to 1. \]
\end{lem}

\begin{proof}
We define the binary random variable $C_u$ to denote whether user $u$'s obfuscation sequence $\textbf{Z}_u$, for $u=1,2,\ldots,n$, contains user $1$'s identifying pattern. $C_u = 1$ indicates that $Z_u$ contains user $1$'s identifying pattern, and $C_u = 0$ otherwise.

$N^{(n)}$ is the total number of users who have the same pattern as user $1$'s identifying pattern; thus $N^{(n)}= \sum_{u=1}^nC_u$.  Recall that $\mathbb{P}\left(\mathcal{B}_u'\right) = \mathbb{P}(\text{user } u\text{ has pattern of user } 1)   \geq \frac{1}{n^{1-\beta}}$; hence,
\begin{align}
\mathbb{E}\left[N^{(n)}\right] =\Pi_{u=1}^{n} \mathbb{E}\left[C_u\right]\geq n \frac{1}{n^{1-\beta}} = n^\beta.\ \
\label{eqN1}
\end{align}
On the other hand, by employing the Chernoff bound, we have
\begin{align}
\mathbb{P}\left(N^{(n)} \leq (1-\delta) \mathbb{E}\left[N^{(n)}\right] \right) &\leq \exp\left(-\frac{\delta{}^2}{2}\mathbb{E}\left[N^{(n)}\right]\right).\ \
\label{eqN2}
\end{align}
Now if we assume $\delta =0.5$, by (\ref{eqN1}) and (\ref{eqN2}), we can conclude
\begin{align}
\mathbb{P}\left(N^{(n)} \leq \frac{n^\beta}{2} \right) &\leq \mathbb{P}\left(N^{(n)} \leq \frac{ \mathbb{E}\left[N^{(n)}\right]}{2} \right)\\
& \leq \exp\left(-\frac{\mathbb{E}\left[N^{(n)}\right]}{8}\right)\\
& \leq \exp\left(-\frac{n^\beta}{8}\right)\to 0,
\label{eqN3}
\end{align}
as $n \to \infty$. As a result, as $n$ becomes large, $\mathbb{P}\left(N^{(n)} \geq \frac{n^\beta}{2} \right) \to 1.$
In other words, the total number of users who have the same pattern as user $1$'s identifying pattern goes to infinity.
\end{proof}

The following theorem shows that by using the proposed SL-SBU technique, we can indeed achieve a privacy guarantee against pattern matching attacks while employing a small obfuscation probability, in fact with $p_{\text{obf}} \rightarrow 0$ as $n \rightarrow \infty$. As motivation, note that in many practical scenarios the size of the alphabet $r$ could be very large and indeed can scale with $n$, the number of users. For example, consider a scenario where the data shows the location of users in an area of interest such as a town or a neighborhood within a city. Assuming a certain level of granularity in the location data, the number of possible locations $r$ and number of users $n$ become larger as the considered area becomes larger.  In such scenarios, it makes sense to write $r=r(n)$ to explicitly denote that $r$ can change as a function of $n$. 

%Indeed, these are the situations that pattern matching attacks could potentially be dangerous- as the alphabet size increases, it becomes exponentially less likely that an identifying pattern of a user is observed within other users' data. Thus, the following theorem focuses on the case where $r=r(n)$ and aims for achieving a guarantee against pattern matching attacks with a vanishing $p_{\text{obf}}$.

\begin{thm}
For the SL-SBU method, let $l>1$ and $h \geq 1$ be fixed. Choose $0<\beta <1$, and define $d(n) = {m(n)}{n^{-\frac{1-\beta}{l-1}}}$. If $[d(n)n^{\theta}]^{\frac{1}{l}} \leq r(n) \leq [d(n)n^{\theta l}]^{\frac{1}{l}}$ for some $0<\theta < \frac{1-\beta}{l-1}$, then by choosing $p_{\text{obf}} = b_n = {n^{-\frac{1-\beta}{l-1}+\theta}}$, and $ \liminf_{n \rightarrow \infty} mb_n \geq 9$, we have
\[ \mathbb{P}\left(\mathcal{B}_u'\right) \geq \frac{c}{n^{1-\beta}},\]
for some constant $c=c(h,l)$.  
\label{thm4}
\end{thm}
\begin{proof}
First note that the assumptions $b_n = {n^{-\frac{1-\beta}{l-1}+\theta}}$, and $ \liminf_{n \rightarrow \infty} mb_n \geq 9$ imply that $m(n) \rightarrow \infty$ as $n \rightarrow \infty$. Recall from Theorem~$2$ that: 
\begin{align}
 & \mathbb{P}\left(\mathcal{B}_u'\right) \geq  \frac{\left( 1-\left( 1-p_{\text{obf}}\right)^h \right)^{(l-1)}}{r^l} \sum_{\alpha=0}^{\min \left\{(r^l-1), \big{\lfloor} Gp_{\text{obf}}\big{\rfloor} \right\} }\hspace{-30 pt} 1-\exp \left(-\frac{\delta_\alpha'^2}{2} Gp_{\text{obf}} \right),\ \
 \label{eq:second-term}
\end{align}
where
 \[G=m-h(l-1), \ \ \delta_\alpha'=1-\frac{\alpha}{Gp_{\text{obf}}}, \text{ for } \alpha=0,1,\cdots, r^l-1. \]

Note that for any $\tau\in{}\mathbb{R}, 1-\tau \leq e^{-\tau}$; thus, with $\tau=b_n$,
\begin{align}
 (1-b_n)^h \leq  e^{-b_nh},\quad \text{for } b_n\in{}\mathbb{R}. \label{theo3_3}
\end{align}
In addition, for any $0 \leq \upsilon \leq 1$,  $1 - e^{-\upsilon} \geq \frac{\upsilon}{2}$; thus, with $\upsilon =b_nh$,
\begin{align}
1 - e^{-b_nh} \geq \frac{b_nh}{2}, \quad \text{for } 0\leq b_nh \leq 1. \label{theo3_4}
\end{align}
Now, by (\ref{theo3_3}) and (\ref{theo3_4}), we can conclude:
\begin{align*}
1-(1-b_n)^h &\geq{} 1- e^{-b_nh} \geq \frac{b_nh}{2}, \text{ for } 0\leq b_nh\leq 1.
\end{align*}
Note that $h$ is a constant, and $b_n \rightarrow 0$ as $n \to \infty$.  As a result: 
\begin{align}
\frac{\left(1-(1-b_n)^h\right)^{l-1}}{r^l} \geq{} \left(\frac{b_nh}{2}\right)^{l-1}\cdot{}\frac{1}{r^l} = \left(\frac{h}{2}\right)^{l-1} \frac{b_n^{l-1}}{r^l}.\ \
\label{eqNaz}
\end{align}
From the statement of the theorem, $[d(n)n^{\theta}]^{\frac{1}{l}} \leq r(n) \leq [d(n)n^{\theta l}]^{\frac{1}{l}}$ for some $0<\theta < \frac{1-\beta}{l-1}$; thus,
\begin{align*}
r^l \geq d(n) n^{\theta} = {m}{n^{-\frac{1-\beta}{l-1}+\theta}} = mb_n,
\end{align*}
as a result,
\begin{align}
r^l-1 \geq  mb_n-1.\ \
\end{align}
Since $G = m - h(l-1) \leq m$, we have:
\begin{align}
r^l-1 \geq  Gb_n-1.\ \
\label{eqNaz2}
\end{align}
Also, 
\begin{align}
Gb_n-1 \leq \big{\lfloor} Gp_{\text{obf}}\big{\rfloor} = \lfloor Gb_n \rfloor.\ \
\label{eqNaz3}
\end{align}
Thus, by (\ref{eqNaz2}) and (\ref{eqNaz3}),
\begin{align*}
Gb_n-1 \leq \min  \left\{(r^l-1), \big{\lfloor} Gp_{\text{obf}}\big{\rfloor} \right\} \leq \lfloor Gb_n \rfloor.
\end{align*}
The above equation can be used to obtain a lower bound for the second term on the right side of (\ref{eq:second-term}):
\begin{align}
& \sum_{\alpha=0}^{\min\left\{(r^l-1), \big{\lfloor}{}Gp_{\text{obf}} \big{\rfloor}
\right\}} \left\{ 1 - \exp\left(-\frac{1}{2}\delta_\alpha'^2Gp_{\text{obf}}\right) \right\} \nonumber \\
& =  \min  \left\{(r^l-1), \big{\lfloor} Gp_{\text{obf}}\big{\rfloor} \right\} + 1 \nonumber \\
& - \sum_{\alpha=0}^{\min\left\{r^l-1, \big{\lfloor} Gp_{\text{obf}}\big{\rfloor}
\right\}}  \exp\left(-\frac{1}{2}\delta_\alpha'^2Gp_{\text{obf}}\right) \nonumber \\
& \geq Gb_n - \sum_{\alpha=0}^{ \lfloor Gb_n \rfloor} \exp\left(-\frac{1}{2}\delta_\alpha'^2Gb_n\right). \label{theo4_deri}
\end{align}

Now since $G = m - h(l-1)$, $h$ and $l$ are constants, and $b_n \rightarrow 0$ as $n \to \infty$, 
\begin{align*}
\liminf_{n \rightarrow \infty} Gb_n &= \liminf_{n \rightarrow \infty} (mb_n-h(l-1)b_n)\\
& = \liminf_{n \rightarrow \infty} mb_n \geq 9.
\end{align*}
Thus, for large enough $n$, $Gb_n \geq 8$.
Note that for $\alpha = 0, \ldots, \lfloor\frac{Gb_n}{2}\rfloor$: 
\begin{align*}
\delta{}'_\alpha = 1 - \frac{\alpha{}}{Gb_n} \geq 1 - \frac{1}{2} = \frac{1}{2};
\end{align*}
thus, for $\alpha = 0, \ldots, \lfloor\frac{Gb_n}{2}\rfloor$,
\begin{align}
\exp\left(-\frac{1}{2}\delta_\alpha'^2Gb_n\right) \leq \exp\left(-\frac{Gb_n}{8}\right) \leq  \exp\left(-1\right).\ \
\label{eqNaz5}
\end{align}

On the other hand,for $\alpha = \lfloor\frac{Gb_n}{2}\rfloor+1, \ldots, \lfloor{Gb_n}\rfloor$: 
\begin{align*}
\delta{}'_\alpha = 1 - \frac{\alpha{}}{Gb_n} \geq 0,
\end{align*}
and as a result, for $\alpha = \lfloor\frac{Gb_n}{2}\rfloor+1, \ldots, \lfloor{Gb_n}\rfloor$, 
\begin{align}
\exp\left(-\frac{1}{2}\delta_\alpha'^2Gb_n\right) \leq 1.\ \
\label{eqNaz6}
\end{align}
%
%So we have:
%\begin{align}
%\exp\left(-\frac{1}{2}\delta_\alpha'^2Gb_n\right) \leq \exp\left(-\frac{3}{4}\right), \quad \text{for } \alpha = 0, \ldots, \lfloor\frac{Gb_n}{2}\rfloor\label{theo4_bound}
%\end{align}
Now by (\ref{eqNaz5}) and (\ref{eqNaz6}), we conclude:
\begin{align}
& \sum_{\alpha=0}^{ \lfloor Gb_n \rfloor } \exp\left(-\frac{1}{2}\delta_\alpha'^2Gb_n\right)  \leq \left(\frac{Gb_n}{2}+1\right) \exp\left(-1\right)  + \left(\frac{Gb_n}{2}+1\right) \times 1 \\
& = \frac{Gb_n}{2} \left(1+ \exp\left(-1\right) + \frac{2 \left(1+ \exp\left(-1\right) \right) }{Gb_n} \right)\\
&\leq \frac{Gb_n}{2} \left(1+ \exp\left(-1\right) + \frac{2 \left(1+ \exp\left(-1\right) \right) }{8} \right)\\
&\leq 0.86 Gb_n.\ \
\label{eqNaz7}
\end{align}
As a result, by (\ref{theo4_deri}) and (\ref{eqNaz7}), for large enough $n$, 
\begin{align*}
& \sum_{\alpha=0}^{\min\left\{(r^l-1), \big{\lfloor}{}Gp_{\text{obf}} \big{\rfloor}
\right\}} \left\{ 1 - \exp\left(-\frac{1}{2}\delta_\alpha'^2Gp_{\text{obf}}\right) \right\}  \\
&\geq .14 Gb_n \geq 0.1  Gb_n.
\end{align*}
%and since for large enough $n$, we also have $Gb_n>8$, we conclude,
%\begin{align*}
%%& \sum_{\alpha=0}^{\min\left\{(r^l-1), \big{\lfloor}{}Gp_{\text{obf}} \big{\rfloor}
%\right\}} \left\{ 1 - %\exp\left(-\frac{1}{2}\delta_\alpha'^2Gp_{\text{obf}}\right) %\right\} \geq \frac{Gb_n}{100},
%\end{align*}
%as $n \to \infty$.
Since $G = m - h(l-1) \leq m$, where $h$ and $l$ are constants, and $m \rightarrow \infty$, we conclude for large enough $n$, $G(n) \geq \frac{m(n)}{2}$,
\begin{align}
& \sum_{\alpha=0}^{\min\left\{(r^l-1), \big{\lfloor}{}Gp_{\text{obf}} \big{\rfloor}
\right\}} \left\{ 1 - \exp\left(-\frac{1}{2}\delta_\alpha'^2Gp_{\text{obf}}\right) \right\} \geq 0.05{mb_n}  .\ \
\label{eq23}
\end{align}
Now, by (\ref{eq:second-term}), (\ref{eqNaz}), and (\ref{eq23}), we conclude that for some constant $c=c(h,l)$, 
\begin{align*}
\mathbb{P}\left(\mathcal{B}_u'\right) \geq \left (\frac{h}{2} \right)^{l-1} \frac{b_n^{l-1}}{r(n)^l}\frac{Gb_n}{10} \geq c \frac{mb_n^l}{r(n)^l}.
\end{align*}
Since $r(n)^l \leq d(n)n^{\theta{}l} = {m}{n^{-\frac{1-\beta}{l-1}}} \times n^{\theta{}l}$, and $b_n ={n^{-\frac{1-\beta}{l-1}+\theta}}$:  
\begin{align*}
\mathbb{P}\left(\mathcal{B}_u'\right) \geq c mb_n^l \times \frac{n^{\frac{1-\beta}{l-1}}}{m\times n^{\theta{}l}} = c\frac{n^{\frac{1-\beta}{l-1}}}{n^{\frac{1-\beta}{l-1}l-\theta{}l}n^{\theta{}l}} = \frac{c}{n^{\frac{1-\beta}{l-1}(l-1)}} = \frac{c}{n^{1-\beta}}.
\end{align*}
\end{proof}

The SL-SBU and i.i.d. obfuscation schemes will be compared extensively via simulation in Section \ref{sec:numerical}.  Here we provide an analytical result to both predict the results of that comparison and provide insight into such. 

\begin{thm}
Suppose the sequence is $[X_1, X_2,\ldots]$ and the pattern is $Q=[q_1, q_2, \ldots, q_l]$.
%(data point size $|\mathcal{R}|=r$, pattern length $l$). 
We say the pattern occurs\footnote{Here ``the pattern occurs'' is the specific case of Definition~$1$ when $h=1$.} in the sequence at time index $t$ if $X_t = q_1$, $X_{t+1} = q_{2}$, $\ldots$ , $X_{t+l-1} = q_l$. We use $T_{\text{SL-SBU}}$ to denote the time index where the pattern first occurs in the SL-SBU obfuscation sequence. Similarly, we use $T_{\text{i.i.d.}}$ to denote the time index where the pattern first occurs in the i.i.d. obfuscation sequence.  Then, the expectation of these times are given by:
\[\mathbb{E}[T_{\text{SL-SBU}}] = \frac{r^l+1}{2},\] 
\[\mathbb{E}[T_{\text{i.i.d.}}] \geq r^l.\] \label{Theo5}
\end{thm}
\begin{proof}
For the SL-SBU obfuscation sequence with length $r^l + l -1$, we can immediately establish the result by the property of the De Bruijn-based sequences of Theorem~$2$:
\begin{align}
 \no \mathbb{E}[T_{\text{SL-SBU}}] &= \frac{1}{r^l} \left[ 1+2+\cdots + r^l  \right]\\
 \no &= \frac{r^l+1}{2}.
\end{align}
Now, we consider the i.i.d. obfuscation sequence, and the corresponding $\mathbb{E}[T_{\text{i.i.d.}}]$. We say that the pattern $Q$ has an overlap of length $l'<l$ if
\[q_1=q_{l-l'+1}, q_2=q_{l-l'+2}, q_{l'}=q_{l}.\]

The largest such $l'$ is called the overlap value of the sequence $Q$ and shown by $l_Q$; thus, for every pattern of length $l$, we have $0 \leq l_Q \leq l-1$. Now the arrival times of pattern $Q$ in the i.i.d. sequence $[X_1, X_2,\ldots]$ can be modeled as a delayed renewal process with an average inter-arrival time $\mu$.  By Blackwell's theorem for delayed renewal processes~\cite{stochastic_processes}, we have
\[ \lim_{t \rightarrow \infty} P(\textrm{Renewal at } t) =\frac{1}{\mu}. \]
Since the sequence is i.i.d., we also have
\[ \lim_{t \rightarrow \infty} P(\textrm{Renewal at } t) = \frac{1}{r^l}.\]
We conclude $\mu=r^l$. Now, consider two cases: $l_Q=0$ and $l_Q>0$. If $l_Q=0$, then 
\[ \mathbb{E}[T_{\text{i.i.d.}}]=\mu=r^l. \]
On the other hand, if $l_Q>0$, let $Q'=[q_1, q_2, \ldots, q_{l_Q}]$. In this case, let also $T_{\text{i.i.d.}}(Q')$ denote the time index where the pattern $Q'$ first occurs in the i.i.d. obfuscation sequence. We have
\[ \mathbb{E}[T_{\text{i.i.d.}}]=\mathbb{E}[T_{\text{i.i.d.}}(Q')]+\mu \geq \mu =r^l. \]
\end{proof}

Finally, we note that the i.i.d. obfuscation approach of \cite{takbiri2018matching} can be readily combined with the techniques proposed here to provide robust privacy simultaneously against both statistical matching and pattern matching attacks.  %The details are included in the Supplementary Material.
%\begin{thm}
%\label{thmNaz} If $\textbf{Z}$ is the obfuscated version of $ \textbf{X}$ after two stages of obfuscation, and $\textbf{Y}$ is the anonymized
%version of $\textbf{Z}$, and:
%\begin{itemize}
%\item The length of the time series data, $m=m(n)$, is arbitrary;
%\item The noise level of the obfuscation method is $\psi_n = \Omega\left(\max \left\{n^{-\frac{1}{r-1}},n^{-\frac{1-\beta}{l-1}}\right\}\right)$
%\end{itemize}
%then, user $1$ has: 
%\begin{itemize}
%    \item\textit{privacy against pattern matching attacks} in any case
%    \item \textit{perfect privacy} if the assumptions about the statistical model of users' data is accurate.
%\end{itemize}
%\end{thm}
%
%\begin{proof}
%See the Supplementary Material.
%\end{proof}

%% file: combination.tex
\section{Combination of i.i.d. Obfuscation and SL-SBU Obfuscation}
To this point, we have employed the SL-SBU obfuscation method to protect users' data against pattern matching attacks while the adversary makes no assumptions about the statistical model of users' data sequences. However, this method has a drawback: as the number of possible values for each user's data points ($r(n)$) increases, it becomes exponentially less likely that an identifying pattern of a user is observed within other users' data; as a result, a pattern matching attack would become a serious threat to users' privacy. 

On the other hand, Takbiri et al.~\cite{takbiri2018matching} considered a strong assumption regarding the statistical model of users' data and introduced a simple i.i.d. obfuscation method in which the samples of the data of each user are reported with error with a certain probability, where that probability itself is generated randomly for each user. In other words, the obfuscated data is obtained by passing the users' data through an $r-$ary symmetric channel with a random error probability. Takbiri et al.~\cite{takbiri2018matching} demonstrated that if the amount of noise level is greater than a critical value, users have perfect privacy against all of the adversary's possible attacks. The definition of perfect privacy is adopted from~\cite{tifs2016}:
 		\begin{define}
		User $u$ has \emph{perfect privacy} at time $k$ if and only if
		\begin{align}%\label{}
		\no \lim\limits_{n\to \infty} \mathbb{I} \left(X_u(k);{\mathbf{Y}}\right) =0,
		\end{align}
		where $\mathbb{I}\left(X_u(k);{\textbf{Y}}\right)$ denotes the mutual information between the data point of user $u$ at time $k$ and the collection of the adversary's observations for all of the users.
	\end{define}

Here, we will combine these two methods of obfuscation in order to benefit from the advantages of both methods and achieve perfect privacy. Note that combining these two techniques does not have any cost asymptotically.
 	
\begin{figure}[h]
	\centering
	\includegraphics[width = 1.0\linewidth]{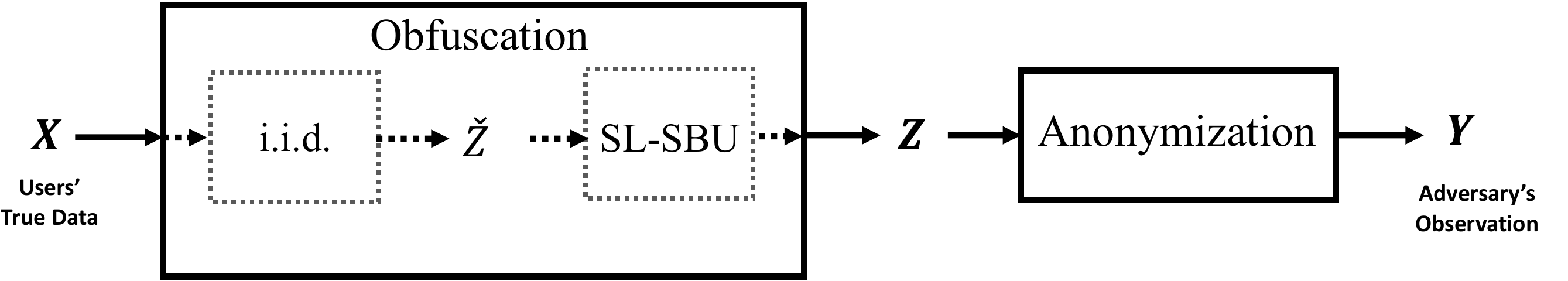}
	\caption{Applying two stages of obfuscation and then anonymization to the users' data points.}
	\label{fig:xyz2}
\end{figure}

As shown in Figure~\ref{fig:xyz2}, two stages of obfuscations and one stage of anonymization are employed to achieve perfect privacy. Note that the first stage is the same i.i.d. obfuscation technique given in ~\cite[Theorem 2]{takbiri2018matching}, and the second stage of obfuscation is the SL-SBU method introduced previously.   Thus, in Figure~\ref{fig:xyz2}, $\check{Z}_u(k)$ shows the (reported) data point of user $u$ at time $k$ after applying the first stage of obfuscation with the noise level equal to $a_n=\Omega\left({n^{-\frac{1}{r-1}}}\right)$, and $Z_u(k)$ shows the (reported) data point of user $u$ at time $k$ after applying the second stage of obfuscation with the noise level equal to $b_n=\Omega \left(n^{-\frac{1-\beta}{l-1}}\right)$.  Define the noise level of a two-stage obfuscation scheme with independent obfuscation probabilities $a_n$ and $b_n$ as $\psi_n=a_n+b_n-a_n b_n$.  We then have the following result.

% HERE (DG):  Need to define the noise level of the joint obfuscation scheme, not in the proof.

\begin{thm}
\label{thmNaz} \textup{If $\textbf{Z}$ is the obfuscated version of $ \textbf{X}$ after two stages of obfuscation, and $\textbf{Y}$ is the anonymized
version of $\textbf{Z}$, and:
\begin{itemize}
\item The length of the time series data, $m=m(n)$, is arbitrary.
\item The noise level of the obfuscation method is $\psi_n = \Omega\left(\max \left\{n^{-\frac{1}{r-1}},n^{-\frac{1-\beta}{l-1}}\right\}\right)$.
\end{itemize}
Then, user $1$ has: 
\begin{itemize}
    \item\textit{privacy against pattern matching attacks} in any case;
    \item \textit{perfect privacy} if the assumptions about the statistical model of users' data is accurate.
\end{itemize}}
\end{thm}
\begin{proof}
First, we show that if the assumptions for the statistical model of users' data is accurate, users will have perfect privacy. We employ a noise level for the first stage of obfuscation equal to $a_n=\Omega\left({n^{-\frac{1}{r-1}}}\right)$, and the noise level for the second stage of obfuscation equal to $b_n=\Omega \left(n^{-\frac{1-\beta}{l-1}}\right)$. Using the definition for the noise level for two-stage obfuscation given before the theorem statement, the noise level of the combined obfuscation mechanism is:
\begin{align}
\psi_n=\Omega\left(\max \{a_n,b_n\}\right)=\Omega\left(\max \left\{n^{-\frac{1}{r-1}},n^{-\frac{1-\beta}{l-1}}\right\}\right),
\end{align}
as $n \to \infty$. From ~\cite[Theorem 2]{takbiri2018matching}: if $a_n$ is significantly larger than $\frac{1}{n^{r-1}}$,
then all users have perfect privacy independent of the value of m(n). Now, since $\psi_n \geq a_n$, by employing a noise value equal to $ \psi_n=\Omega\left(\max \left\{n^{-\frac{1}{r-1}},n^{-\frac{1-\beta}{l-1}}\right\}\right)$, all users achieve perfect privacy independent of the value of m(n). In other words, as $n \to \infty$, $\mathbb{I} \left(X_u(k);{\mathbf{Y}}\right) =0.$

If the assumption regarding the statistical model of users' data is accurate or not, Theorem~$3$ establishes that users would have privacy against pattern matching attack due to the second stage of our obfuscation method.
\end{proof}

%% file: data_dependent.tex
\section{Data-Dependent Obfuscation}
\label{sec:data-dependent}

The obfuscation techniques proposed in Section \ref{sec:simple-ppm} are independent of the user data, as would be appropriate for real-time operation on non-buffered data as discussed in Section \ref{intro}.  However, as also discussed in Section \ref{intro}, there are scenarios such as image processing where the entire data sequence might be known to the PPM.  To exploit such, we employ opportunistic \textit{superstring} creation, which we refer to as \emph{data-dependent obfuscation (DDO)}\footnote{Note that the sequences developed might not technically be {\em superstrings}, as defined formally in Section \ref{sec:simple-ppm}, but, since the sequences are employed in a similar fashion to the {\em superstrings} of Section \ref{sec:simple-ppm}, we employ the same term to avoid confusion.}. The key point here is to choose obfuscated values $a_u(j)$ in an opportunistic fashion; that is, at each point, the element $a_u(j)$ in the \textit{superstring} is chosen based on the realized obfuscated sequence so far, with the goal of choosing $a_u(j)$ in a way to maximize the number of distinct patterns in the obfuscated sequence of user $u$. Figure \ref{fig:OSC} shows the structure of the DDO algorithm.

%{\bf \textcolor{red}{BG: for Hossein: DOO algorithms might not necessarily create a \textit{superstring} which contains all the possible pattern. We might change ``\textit{superstring}'' to a different word.}}

 \begin{figure}[h]%{wrapfigure}{r}{0.5\linewidth}
 	%\vspace{-5pt}
 	\centering
 	\includegraphics[width=\textwidth]{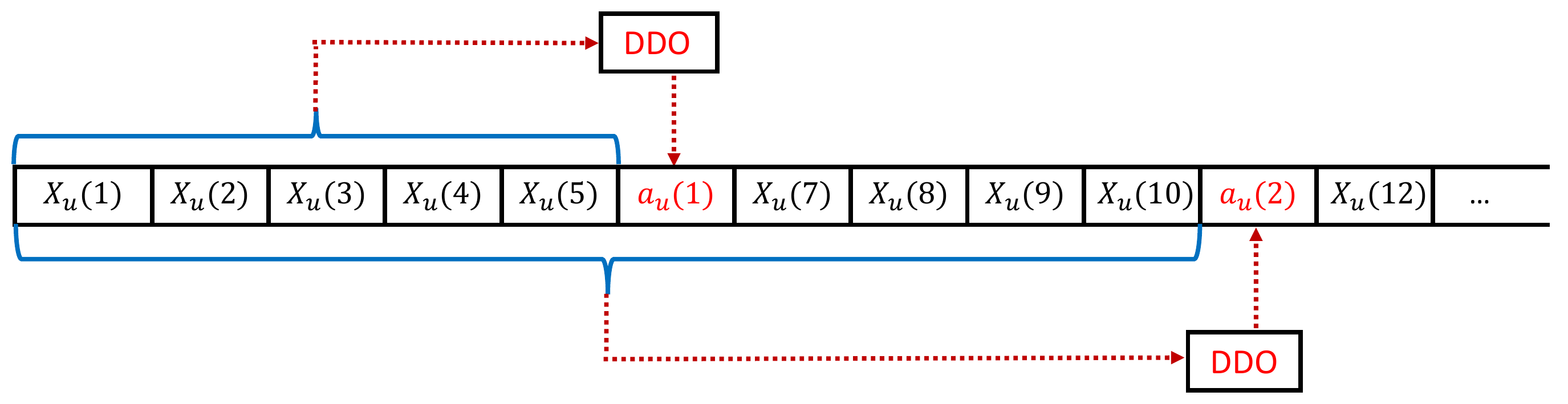}
 	\vspace{-7pt}
 	\caption{Data-dependent Obfuscation: $a_u(j)$ is chosen by the DDO algorithm based on the realized obfuscated sequence so far.}
 	\label{fig:OSC}
 	\vspace{-5pt}
 \end{figure}%{wrapfigure}

%The key challenge would be efficiency. That is to have an algorithm with low-complexity both in terms of time and space. In what follows, we propose several algorithms with low-complexity that provide good performance.

\subsection{DDO Algorithms to Thwart Pattern Matching Attacks}
%We take the reasonable approach of designing DDO algorithms which result in the ``most uniform'' distribution of patterns in the users.  The challenge is in the time and space complexity of the resulting designs.  

First we formulate solutions in a general setting for arbitrary pattern length $l$. Then we design three examples of DDO algorithms for specific values of $l$. Indeed, some practical pattern matching attacks are based on patterns with small length for identification or classification~\cite{christodoulakis2004pattern} \cite{cheng2013pattern}, and, as validated in our simulation results later, low-order DDOs can also work well for larger $l$ than they were designed. 

In what follows we drop the subscript $u$ to simplify notation.  Hence, let $X(k)$ be the data point at index $k$ for an arbitrary user and $Z(k)$ be its obfuscated version.

%As before, let $X_u(k)$ be the data point at time index $k$ for user $u$. In what follows we drop the subscript $u$ to simplify notation. Thus, $X(k)$ is the data sequence for a random user, and $Z(k)$ is the obfuscated data point at time index $k$. The goal is for the obfuscated sequence $Z$ to have a high probability of having any chosen pattern (the identifying pattern of user $1$). Let us provide some definitions below.

\begin{define}
For a \textit{pattern} $Q$ of length $l$, $N^l_k(Q)$ is the total number of times the \textit{pattern} $Q$ has been observed as a \textit{pattern} (Definition \ref{def:pattern}) up to and including time $k$. That is, it is the number of times the \textit{pattern} has been observed in
\[Z(1),Z(2), \cdots, Z(k).\]
A pattern distribution $N^l_k$ for the sequence $Z(1),Z(2), \cdots, Z(k)$ is the collection of values $N^l_k(Q)$ across all \textit{pattern}s $Q$ of length $l$.
\end{define}

In the special case of $l=1$, $N^1_k(i)$ is the number of times the value $i$ has appeared in $Z(1),Z(2), \cdots, Z(k)$.

\begin{define}
A DDO algorithm of order $l$ is a mapping from the set of all pattern distributions to the set of probability distributions over data point set $\mathcal{R}=\{0,1, \ldots, r-1\}$. This probability distribution, denoted by $P_{\text{DDO}}$, provides the probability of obfuscating $X(k+1)$ to the values $0,1, \ldots, r-1$, given that we are performing obfuscation on a given data sample.  
\end{define}

The simplest DDO algorithm, which we call \textit{Least-Observed Value (LOV)}, works as follows: to obfuscate $X(k+1)$, choose one of the values not present in $Z(1),Z(2), \cdots, Z(k)$, uniformly at random. If all values have been observed, we obfuscate the data points with a value drawn uniformly at random from $\mathcal{R}$.  To execute the algorithm, we only 
need to keep the subset of $\mathcal{R}$ containing the values that have not been observed to this point and choose one of them at random for obfuscation.  Denote $\mathbb{P}_{\text{LOV}}$ as the probability that the obfuscated sequence has user $1$'s identifying pattern after applying the LOV algorithm.  For $l=1$: \[ \mathbb{P}_{\text{LOV}} \geq \sum_{k=0}^{r-1} {m \choose k} p_{\text{obf}}^k(1-p_{\text{obf}})^{m-k} \left(\frac{k}{r}\right) + \sum_{k=r}^{m} {m \choose k} p_{\text{obf}}^k(1-p_{\text{obf}})^{m-k},\] where $m$ is the length of the sequence.

The second DDO Algorithm, which we term \textit{Probabilistic Least-Observed Value (PLOV)}, is in some sense a generalization of the LOV algorithm that introduces more randomness in the operation. The intuition behind the obfuscation of PLOV is to give higher probability to the values that have appeared less so far. Specifically, at time $k$, define:

\[ \tilde{q_i}= \left(\frac{N^1_k(i)}{k} \right)^{\gamma},\]
where $0<\gamma$ is a design parameter. A typical value is $\gamma=\frac{1}{10}$. Now let
\[ q_i = \frac{\tilde{q_i}}{\sum_{j=1}^{r} \tilde{q_j}} \]
\[ q_{max} = \max \{ q_i, i=0,1,\cdots, r-1\}, \]
\[ q_{min} = \min \{ q_i, i=0,1,\cdots, r-1\}, \]
and choose
\[ b \leq \min \left( \frac{1}{rq_{max}-1}, \frac{r-1}{1-rq_{min}}\right) .\]

For example, we set $b =0.99 \min \left( \frac{1}{rq_{max}-1}, \frac{r-1}{1-rq_{min}}\right)$ in our experiments. The obfuscation probabilities are given by
\[p_i=\frac{1+b}{r}-b q_i, \ \ \ i=0,1,2, \cdots, r-1.\]
where $p_i$ is the conditional probability of obfuscating to $i$ ($Z_u(k+1)=i$), given that we are obfuscating $X_u(k+1)$.

The third DDO algorithm is termed \textit{Make a New Pattern (MANP)}, which chooses a value that completes as many patterns as possible with length $l$ that have not been observed to this point. Specifically, we choose the value $l=2$. %Thus it needs $O(r^2)$ space for storing the set which contains all possible patterns. Also, given that we are obfuscating, we need $\log_2r^2 = O(\log_2r)$ operations for checking (searching) if the created pattern exists in the previous obfuscated sequence.

Define $\mathfrak{P}_k$ as the total number of distinct patterns of length $l=2$ observed in the obfuscated sequence until time $k$, (i.e., in $Z(1), Z(2), \cdots, Z(k)$). 
%\textcolor{red}{That is, 
%\[\mathfrak{P}_k = \sum\limits_{Q\in \mathcal{R}^2}\big{\lfloor}\frac{\text{sgn}(N^2_k(Q))+1}{2}\big{\rfloor}.\] BG: For Hossein: to review.}
Thus, $\mathfrak{P}_1=0$ and $\mathfrak{P}_2=1$.  Also, for $i \in \mathcal{R}$, define $\mathfrak{P}_{k+1}(i)$ as the value of $\mathfrak{P}_{k+1}$ given that $Z(k+1)=i$. Given we are obfuscating at time $k+1$, choose
\[Z(k+1) = \argmax{}\mathfrak{P}_{k+1}(i), \quad i \in \mathcal{R}.\]
These three DDO algorithms (LOV, PLOV, MANP) will be simulated in the next section.

%% file: complexity.tex
\section{Complexity Analysis}
The time complexity and space complexity for each obfuscation algorithm (DIO and DDO) for each user's data sequence with length $m$ are shown in Table~\ref{complexity} based on the following assumptions:
\begin{itemize}
%\item We assume it takes constant time complexity $O(1)$ for the random number generator and for shifting of the \textit{superstring} (by using a \textit{linked list} data structure). 
\item We assume searching/insertion time complexity for a specific element in/to a set with size $N$ is $O(\log N)$. %(by using a \textit{balanced binary search tree} data structure).
\item We assume the sorting algorithm takes $O(N\log N)$ time complexity for an array with size $N$.
\end{itemize}
The time complexity for each of the two DIO algorithms is $O(m)$, since each obfuscates each data point based on the obfuscation sequence, which takes constant time. For LOV, it takes $O(\log r)$ time to check (search) if the current obfuscated data point is (or is not) a member of the letter set which have been seen before (with worst-case size $r$). For MANP, it needs to search in the set of patterns which have been seen before (with worst-case size $O(r^2)$ and search time $O(\log r^2) = O(\log r)$) for each candidate letter (with $r$ different choices in the data point set $\mathcal{R}$, $O(r\log r)$ in total). And it also takes $O(r \log r)$ for sorting the candidate letters by the order of their achievable patterns if used for obfuscating at the current data point. Lastly, it takes $O(h\log r^2) = O(h\log r)$ time complexity for inserting the new pattern list into the set of patterns previously observed after obfuscation (worst-case size $h$). For PLOV, it takes $O(r)$ complexity for each obfuscation operation due to the summation of the elements of vector $\tilde{q}$ with size $r$. For space complexity, the i.i.d., LOV and PLOV method each take $O(r)$ space for storing the data point set.  For the SL-SBU method, it takes $O(r^l)$ space to store the obfuscation sequence (De Bruijn sequence) with length $r^l$ (the number of all possible patterns); for MANP, it takes $O(r^2)$ space for storing the set of patterns which have been seen before (with worst-case size $r^2$) and $O(hr)$ space for storing the counting table for each letter's contribution to create new patterns (maximum size for each letter is equal to $h$, $O(hr)$ in total).

\begin{table}[h]
%\centering
\caption{Time complexity and space complexity of each obfuscation method for each user's data sequence with length $m$.}
\begin{tabular}{|c|c|c|}
\hline
 DIO algorithms & i.i.d. & SL-SBU \\ \hline

 time complexity & $O(m)$ & $O(m)$ \\ \hline

 space complexity & $O(r)$ & $O(r^l)$ \\
\hline
\end{tabular}
\label{complexity}
\end{table}

\begin{table}[h]
%\centering
\begin{tabular}{|c|c|c|c|}
\hline
 DDO algorithms & LOV & PLOV & MANP \\ \hline

 time complexity & $O(m\log r)$ & $O(m r)$ & $O(m \cdot \max\{r,h\} \cdot \log r)$ \\ \hline

 space complexity & $O(r)$ & $O(r)$  & $O(r \cdot \max\{r,h\})$ \\
\hline
\end{tabular}
\end{table}

%% file: numerical_results.tex
\section{Numerical Results and Validation}
\label{sec:numerical}
We evaluate the performance of the proposed \textit{superstring-based obfuscation} (SBU) methods and the three \textit{data-dependent obfuscation} (DDO) algorithms on synthetic i.i.d. data sequences and on sequences from the Reality Mining dataset.  The data points in the i.i.d. data sequence for each user are drawn independently and identically from the data point set $\mathcal{R}$.  Reality Mining is a dataset released by the MIT Media Laboratory which tracks a group of $106$ (anonymized) mobile phone users~\cite{eagle2009inferring}.  The Reality Mining dataset contains traces of users' associated cell tower IDs across time. Here we further sample the data traces with sampling interval at least $10$ minutes to avoid significant data point repetition.

\subsection{Numerical Validation of Theorem 4}
First, we verify Theorem~$4$ numerically, as follows.  For each iteration, a pattern with length $l$ is generated at random, which means each letter in the pattern is independently and identically drawn from the data point alphabet $\mathcal{R}$.  Then, we generate the SL-SBU obfuscation and the i.i.d. obfuscation sequences randomly, and we record $Y$ and $X$ as the indices for the start of the first occurrence of the pattern in the SL-SBU obfuscation sequence and in the i.i.d. obfuscation sequence, respectively.  The numerical results are shown in Table~\ref{table3}.  As anticipated, the expectation of the location index where the pattern first occurs in the SL-SBU obfuscation sequence is approximately one-half of that of where it first occurs for the i.i.d. obfuscation sequence, which coincides with Theorem~$4$.  And, as expected, the likelihood of the event that $X$ is larger than $Y$ is greater than $0.5$.

\begin{table}[h]
\centering
\caption{Numerical validation of Theorem~$4$. $X$ and $Y$ denote the indices for the start of a randomly-generated pattern's first occurrence in the i.i.d. obfuscation sequence and SL-SBU obfuscation sequence, respectively.}
\begin{tabular}{{|}cccc{|}c{|}c{|}c{|}}
\hline
& $r$ & $l$ &$r^l$ &$\mathbb{E}(X)$ &$\mathbb{E}(Y)$ & $P(X>Y)$ \\ \hline

& 10 & 2  &100 & 100.52 & 51.45  & 0.6276  \\

& 20 & 2  &400 & 401.88 & 201.17 & 0.6315 \\

& 30 & 2  & 900 &902.41 & 450.55  & 0.6321   \\

& 40 & 2  & 1600 &1609.46 &802.11 & 0.6314    \\

& 50 & 2  & 2500 & 2515.09 & 1252.04 & 0.6329   \\ \hline \hline

& 10 & 3  &1000 & 1002.65  & 503.02  & 0.6303 \\

& 20 & 3  &8000  &8007.31  & 4008.64  & 0.6327 \\

& 30 & 3  &27000 & 27006.66 & 13503.91 & 0.6315 \\

& 40 & 3  &64000 & 64008.32 & 32003.53  & 0.6303 \\

& 50 & 3  &125000 & 125000.63  & 62502.33  & 0.6321 \\

%& 10 & 4  &10000 &10015.2810  & 5006.5497  & 0.6325\\ %more iterations

%& 15 & 4  &50625 & 50632.1924 & 25315.5072 & 0.6305 \\      %more iterations

%& 20 & 4  &160000 & 160016.5923  & 80003.8943  & 0.6312\\  %more iterations
\hline
\end{tabular}
\label{table3}
\end{table}

\subsection{Evaluation for SBU Obfuscation}
We consider the numerical evaluation of the achievable lower bounds, as given in Theorem~\ref{thm:simple-ppm-privacy-guarantee} and Theorem~\ref{second-thm} (SL-SBU, optimized version), for the fraction of sequences that contain a potentially identifying pattern of user $1$ when using the proposed SBU obfuscation approach.  We use $\epsilon$ and $\epsilon'$ to denote these two lower bounds, respectively, and the results are shown in Table~\ref{table2}.  Note that these are deterministic numerical evaluations of the bounds in Theorem~\ref{thm:simple-ppm-privacy-guarantee} and Theorem~\ref{second-thm}.  Hence, the results show that the proposed PPMs will result in a non-zero percentage of the user set $\mathcal{U}$ that have any potentially identifying pattern of user $1$ in their obfuscated sequences with high probability, and we can observe that the SL-SBU obfuscation sequence has a higher lower bound than the regular SBU obfuscation sequence.  As expected, increasing the data sequence length $m$ or the obfuscation noise level $p_{\text{obf}}$ will increase the chance of observing the pattern in the obfuscated sequences for both methods.  The advantage of the SL-SBU approach over the regular SBU approach becomes more significant as longer sequences are considered.

\begin{table}[h]
\centering
\caption{Numerical evaluation of the lower bounds of Theorem~\ref{thm:simple-ppm-privacy-guarantee} ($\epsilon$) and Theorem~\ref{second-thm} ($\epsilon'$) for the percentage of sequences that contain a user's identifying pattern when the proposed SBU obfuscation approaches are employed.} 
\begin{tabular}{{|}ccccc{|}c{|}c{|}}
\hline
$m$ & $r$ & $l$ & $h$ & $p_{\text{obf}}$ &lower bound $\epsilon$ &lower bound $\epsilon'$ \\
\hline
1000 & 20 & 3 & 10 & 10\% & 0.15\% & 0.45\% \\

1000 & 20 & 3 & 8 & 10\% & 0.12\% & 0.35\% \\

1000 & 20 & 3 & 10 & 15\% & 0.36\% & 1.06\% \\

1000 & 20 & 3 & 10 & 30\% & 1.07\% & 3.22\% \\

4000 & 20 & 3 & 10 & 10\% & 0.66\% & 1.98\% \\

10000 & 20 & 3 & 10 & 10\% & 1.69\% & 5.08\% \\ \hline \hline

% 10000 & 20 & 3 & 10 & 80\% & 32.80\% & 98.40\% \\ 

1000 & 20 & 2 & 10 & 10\% & 7.12\% & 14.17\%\\

1000 & 20 & 2 & 8 & 10\% & 6.24\% & 12.41\%\\

1000 & 20 & 2 & 10 & 15\% & 13.47\% & 26.84\%\\

1000 & 20 & 2 & 10 & 30\% & 33.57\% & 67.02\%\\

% 1000 & 20 & 2 & 10 & 40\% & 46.22\% & 92.32\%\\

2000 & 20 & 2 & 10 & 10\% & 14.84\% & 29.60\%\\

4000 & 20 & 2 & 10 & 10\% & 30.52\% & 60.97\%\\

%3000 & 16 & 2 & 10 & 10\% & 35.27\% & 65.05\%\\

%3000 & 16 & 2 & 10 & 15\% & 66.18\% & 80.31\%\\

%5000 & 30 & 2 & 10 & 10\% & 17.04\% & 34.04\%\\

%10000 & 30 & 2 & 10 & 10\% & 34.71\% & 65.12\%\\

%15000 & 12 & 3 & 10 & 10\% & 11.86\% & 35.57\%\\

%15000 & 15 & 3 & 10 & 10\% & 6.07\% & 18.21\%\\
\hline
\end{tabular}

\label{table2}
\end{table}

Next we test the effectiveness of the i.i.d. obfuscation and SL-SBU obfuscation approaches on synthetic i.i.d. sequences. We believe the i.i.d. sequences yield the worst case scenario: there is no dependency between any two consecutive data points that would lead to common subsequences of potentially identifying patterns being likely to be shared across users.  Furthermore, to consider (pessimistically) only patterns that are inserted via our obfuscation method (eliminating the possibility that a user trace already has the desired pattern), we make certain that a dataset $\mathcal{R}$ with size $r$ has a unique sequence by assigning user $1$ a unique pattern and drawing other users' sequences from the subset of the dataset with size $(r-l)$; for instance, if the pattern length is $l=3$, we insert a pattern $[r-2, r-1, r]$ into user $1$'s sequence at a random place for uniqueness.  We then follow the obfuscation procedure from Section~\ref{sec:framework} for each iteration and calculate the results by averaging the fraction of sequences which contain user $1$'s identifying pattern (by Definition~\ref{def:pattern}) for all iterations. The validation results for different parameter settings are shown in Tables~\ref{table4}, \ref{table5}, \ref{table6}. From the overall results, we can observe that the SL-SBU obfuscation sequence performs better than the i.i.d. obfuscation sequence, as predicted by Theorem~\ref{Theo5}. 
%As noted earlier, this is because the SL-SBU obfuscation sequence takes advantage of the properties of the De Bruijn sequence to provide the shortest length sequence that contains all the possible patterns. 
%In some cases, they might have similar performance if the data sequence is not long enough to hold the obfuscation sequence (e.g., if only a small fraction of the \textit{superstring} of SL-SBU get obfuscated) or if the sequence is long enough to enhance both obfuscation's performance to the same degree. %Even though in some cases, the percentage is low for both obfuscation sequences, from Lemma ~\ref{thm3} and Theorem~\ref{thm4}, we know that there would still exist infinite number of users having the same pattern as user $1$'s as $n$ goes to infinity as long as some parameter conditions are satisfied. A certain percentage can still be achieved even when the obfuscation level is low (e.g., $p_{\text{obf}} = 5\%$).  

\begin{table}[h]
\centering
\caption{Simulation results for the case of i.i.d. data sequences drawn from an alphabet of size $r$, when using an i.i.d. sequence and the SL-SBU sequence for obfuscation: the fraction of sequences which contain user $1$'s identifying pattern ($[r-l+1,\ldots,r-1,r]$) for $h=10$ and $p_{\text{obf}} = 10\%$.}
\begin{tabular}{{|}ccccc{|}c{|}c{|}}
\hline
$m$ & $r$ & $l$ & $h$ & $p_{\text{obf}}$  & fraction (i.i.d.) & fraction (SL-SBU)\\
\hline

%$10^3$ & 10 & 2 & 10 & 10\%  & 0.6031  & 0.9601 \\

%$10^4$ & 10 & 2 & 10 & 10\%  & 0.9999  & 1 \\

%$10^4$ & 10 & 3 & 10 & 10\%  & 0.5910 & 0.7308 \\

%$10^5$ & 10 & 3 & 10 & 10\%  & 0.9998 & 1 \\ \hline \hline

$10^3$ & 20 & 2 & 10 & 10\%  & 0.2185  & 0.7380 \\

$10^4$ & 20 & 2 & 10 & 10\%  & 0.9097 & 1 \\

$10^4$ & 20 & 3 & 10 & 10\%  & 0.1176 & 0.2571\\

$10^5$ & 20 & 3 & 10 & 10\%  & 0.6949 & 0.9598 \\ \hline \hline

$10^3$ & 30 & 2 & 10 & 10\%  & 0.1091 & 0.5853\\

$10^4$ & 30 & 2 & 10 & 10\%  & 0.6624 & 0.9999 \\

$10^5$ & 30 & 3 & 10 & 10\%  & 0.3042 & 0.7656 \\

$10^6$ & 30 & 3 & 10 & 10\%  & 0.9712 & 1 \\ \hline \hline

$10^3$ & 40 & 2 & 10 & 10\%  & 0.0666 & 0.4838 \\

$10^4$ & 40 & 2 & 10 & 10\%  & 0.4621  & 0.9983\\

$10^5$ & 40 & 3 & 10 & 10\%  & 0.1465 & 0.6010\\

$10^6$ & 40 & 3 & 10 & 10\%  & 0.7838 & 0.9999 \\ \hline \hline

$10^3$ & 50 & 2 & 10 & 10\%  & 0.0462  & 0.4142 \\

$10^4$ & 50 & 2 & 10 & 10\%  & 0.3301  & 0.9913  \\

$10^5$ & 50 & 3 & 10 & 10\%  & 0.0808  & 0.4937 \\

$10^6$ & 50 & 3 & 10 & 10\%  & 0.5412 & 0.9994 \\

\hline
\end{tabular}

\label{table4}
\end{table}

\begin{table}[h]
\centering
\caption{Simulation results for the case of i.i.d. data sequences drawn from an alphabet of size $r$, when using an i.i.d. sequence and the SL-SBU sequence for obfuscation: the fraction of sequences which contain user $1$'s identifying pattern ($[r-l+1,\ldots,r-1,r]$) for $h=5$ and $p_{\text{obf}} = 10\%$.}
\begin{tabular}{{|}ccccc{|}c{|}c{|}}
\hline
$m$ & $r$ & $l$ & $h$ & $p_{\text{obf}}$  & fraction (i.i.d.) & fraction (SL-SBU)\\
\hline

%$10^3$ & 10 & 2 & 5 & 10\%  & 0.3863  & 0.6903  \\

%$10^4$ & 10 & 2 & 5 & 10\%  & 0.9918   & 0.9999  \\

%$10^4$ & 10 & 3 & 5 & 10\%  & 0.2172  & 0.2180  \\

%$10^5$ & 10 & 3 & 5 & 10\%  & 0.9053  & 0.9062  \\ \hline \hline

$10^3$ & 20 & 2 & 5 & 10\%  & 0.1223   & 0.3733  \\

$10^4$ & 20 & 2 & 5 & 10\%  & 0.7078  & 0.9932  \\

$10^4$ & 20 & 3 & 5 & 10\%  & 0.0370  & 0.0391  \\

$10^5$ & 20 & 3 & 5 & 10\%  & 0.2683  & 0.2961  \\ \hline \hline

$10^3$ & 30 & 2 & 5 & 10\%  & 0.0607 & 0.2585 \\

$10^4$ & 30 & 2 & 5 & 10\%  & 0.4268  & 0.9587  \\

$10^5$ & 30 & 3 & 5 & 10\%  & 0.0949  & 0.1194  \\

$10^6$ & 30 & 3 & 5 & 10\%  & 0.5891  & 0.7174  \\ \hline \hline

$10^3$ & 40 & 2 & 5 & 10\%  & 0.0383  & 0.1976  \\

$10^4$ & 40 & 2 & 5 & 10\%  & 0.2719  & 0.8846 \\

$10^5$ & 40 & 3 & 5 & 10\%  & 0.0438 & 0.0646 \\

$10^6$ & 40 & 3 & 5 & 10\%  & 0.3271 & 0.4724  \\ \hline \hline

$10^3$ & 50 & 2 & 5 & 10\%  & 0.0277  & 0.1616  \\

$10^4$ & 50 & 2 & 5 & 10\%  & 0.1868  & 0.8020  \\

$10^5$ & 50 & 3 & 5 & 10\%  & 0.0268  & 0.0429 \\

$10^6$ & 50 & 3 & 5 & 10\%  & 0.1840  & 0.3170 \\

\hline
\end{tabular}

\label{table5}
\end{table}

\begin{table}[h]
\centering
\caption{Simulation results for the case of i.i.d. data sequences drawn from an alphabet of size $r$, when using an i.i.d. sequence and the SL-SBU sequence for obfuscation:  the fraction of sequences which contain user $1$'s identifying pattern ($[r-l+1,\ldots,r-1,r]$) for $h=10$, $p_{\text{obf}} = 5\%$.}
\begin{tabular}{{|}ccccc{|}c{|}c{|}}
\hline
$m$ & $r$ & $l$ & $h$ & $p_{\text{obf}}$  & fraction (i.i.d.) & fraction (SL-SBU)\\
\hline

%$10^3$ & 10 & 2 & 10 & 5\%  & 0.2191 & 0.4346 \\

%$10^4$ & 10 & 2 & 10 & 5\%  & 0.9089 & 0.9973 \\

%$10^4$ & 10 & 3 & 10 & 5\%  & 0.1185 & 0.1188 \\

%$10^5$ & 10 & 3 & 10 & 5\%  & 0.6912 & 0.6916 \\ \hline \hline

$10^3$ & 20 & 2 & 10 & 5\%  & 0.0673 & 0.2203 \\

$10^4$ & 20 & 2 & 10 & 5\%  & 0.4622 & 0.9255 \\

$10^4$ & 20 & 3 & 10 & 5\%  & 0.0235 & 0.0259 \\

$10^5$ & 20 & 3 & 10 & 5\%  & 0.1502 & 0.1758 \\ \hline \hline

$10^3$ & 30 & 2 & 10 & 5\%  & 0.0358 & 0.1497 \\

$10^4$ & 30 & 2 & 10 & 5\%  & 0.2454 & 0.7885 \\

$10^5$ & 30 & 3 & 10 & 5\%  & 0.0539 & 0.0758 \\

$10^6$ & 30 & 3 & 10 & 5\%  & 0.3693 & 0.5194 \\ \hline \hline

$10^3$ & 40 & 2 & 10 & 5\%  & 0.0241 & 0.1147 \\

$10^4$ & 40 & 2 & 10 & 5\%  & 0.1509 & 0.6639 \\

$10^5$ & 40 & 3 & 10 & 5\%  & 0.0274 & 0.0444 \\

$10^6$ & 40 & 3 & 10 & 5\%  & 0.1770 & 0.3148 \\ \hline \hline

$10^3$ & 50 & 2 & 10 & 5\%  & 0.0187 & 0.0930 \\

$10^4$ & 50 & 2 & 10 & 5\%  & 0.1025 & 0.5736 \\

$10^5$ & 50 & 3 & 10 & 5\%  & 0.0184 & 0.0314 \\

$10^6$ & 50 & 3 & 10 & 5\%  & 0.1015 & 0.2150 \\

\hline
\end{tabular}

\label{table6}
\end{table}

\subsection{Evaluation of the Data-Dependent and Data-Independent Obfuscation Algorithms}

Next we consider the simulation of the proposed \textit{data-independent obfuscation} (DIO) obfuscation methods (SL-SBU and i.i.d. obfuscation sequences) and the three \textit{data-dependent obfuscation} (DDO) algorithms (LOV, PLOV, MANP) on i.i.d. data sequences and the Reality Mining dataset. %\footnote{The code to reproduce the results present here is available at \url{https://github.com/BillyHub/sequence_obfuscation}.}.  
Recall the three DDO obfuscation algorithms' design: the LOV algorithm chooses the obfuscation value which has not been observed in the user's obfuscated sequence before; the PLOV algorithm selects obfuscating values that have been less-observed in the user's obfuscated sequence with higher probability; MANP chooses the obfuscating letter which completes the most previously unobserved patterns with length $l=2$.

%All the three DDO algorithms have the common property: they choose the obfuscating letter $a_u(j)$ who has the most specific ``potential'' by taking advantage of the available previous obfuscating sequences, though the ``potential'' has different power on different pattern length $l$ and different sequence distribution.

Fig.~\ref{fig:perform_l1_Pobf_iid} and Fig.~\ref{fig:perform_l1_Pobf_RealityMining} show the performance comparison of the obfuscation algorithms on i.i.d. sequences and the Reality Mining dataset, respectively, for the pattern length $l=1$. In this case, the LOV algorithm has the best performance on both i.i.d. sequences and dataset sequences, although all three DDO algorithms achieve very good performance using very low obfuscation probability $p_{\text{obf}}$.  For the two DIO algorithms, their performance is not affected by the data sequence's type.  The SL-SBU obfuscation method's performance is better than the i.i.d. obfuscation method for both i.i.d. sequences and the Reality Mining dataset: the SL-SBU obfuscation method can achieve a fraction of nearly $0.90$ with $p_{\text{obf}} = 0.02$, while the i.i.d. obfuscation method achieves a fraction around $0.60$ fraction with $p_{\text{obf}} = 0.02$.

Fig.~\ref{fig:perform_l2_Pobf_iid} shows the performance comparison of the obfuscation algorithms on i.i.d. sequences with pattern length $l=2$.  In this case, the PLOV algorithm and MANP algorithm have similar performance, while the LOV algorithm has poorer performance.  For the two DIO algorithms, the SL-SBU obfuscation method's performance is better than that of the i.i.d. obfuscation method for both i.i.d. sequences and the Reality Mining dataset.  For the Reality Mining dataset validation results, as shown in Fig.~\ref{fig:perform_l2_Pobf_RealityMining}, MANP's performance is poorer on the realistic dataset compared to the i.i.d. sequence due to the sparseness and repetition of the data points in the real data sequences.  For the two DIO algorithms, their performance is again not affected by the data sequence's type. The SL-SBU obfuscation method's performace is better than the i.i.d. obfuscation method for both i.i.d. sequences and the Reality Mining dataset: the SL-SBU obfuscation method can achieve a fraction of nearly $0.70$ with $p_{\text{obf}} = 0.10$, while the i.i.d. obfuscation method achieves a fraction around $0.20$ with $p_{\text{obf}} = 0.10$.

Fig.~\ref{fig:perform_l3_Pobf} (a) and Fig.~\ref{fig:perform_l3_Pobf} (b) show the performance of the obfuscation algorithms on i.i.d. sequences and the Reality Mining dataset, respectively, when the pattern length $l=3$. In this case, PLOV has the best performance among the three DDO algorithms. For the two DIO algorithms, the SL-SBU obfuscation has the best performance over all of other methods either on i.i.d. sequences (for large enough $p_{\text{obf}}$) or Reality Mining sequences, hence showing its robustness across different data sequence types.  The reason is that the SL-SBU obfuscation can deterministically create all patterns without cooperating with the obfuscated data points and will eventually achieve any pattern if the sequence is long enough or the obfuscation probability $p_{\text{obf}}$ is high enough. In contrast, none of the DDO algorithms are designed specifically for pattern length above $2$.  Overall, we can observe that all three DDO algorithms have poorer performance on the Reality Mining dataset than on the i.i.d. sequence.

\begin{figure}
     \centering
     \begin{subfigure}[b]{0.45\textwidth}
         \centering
         \includegraphics[width=\textwidth]{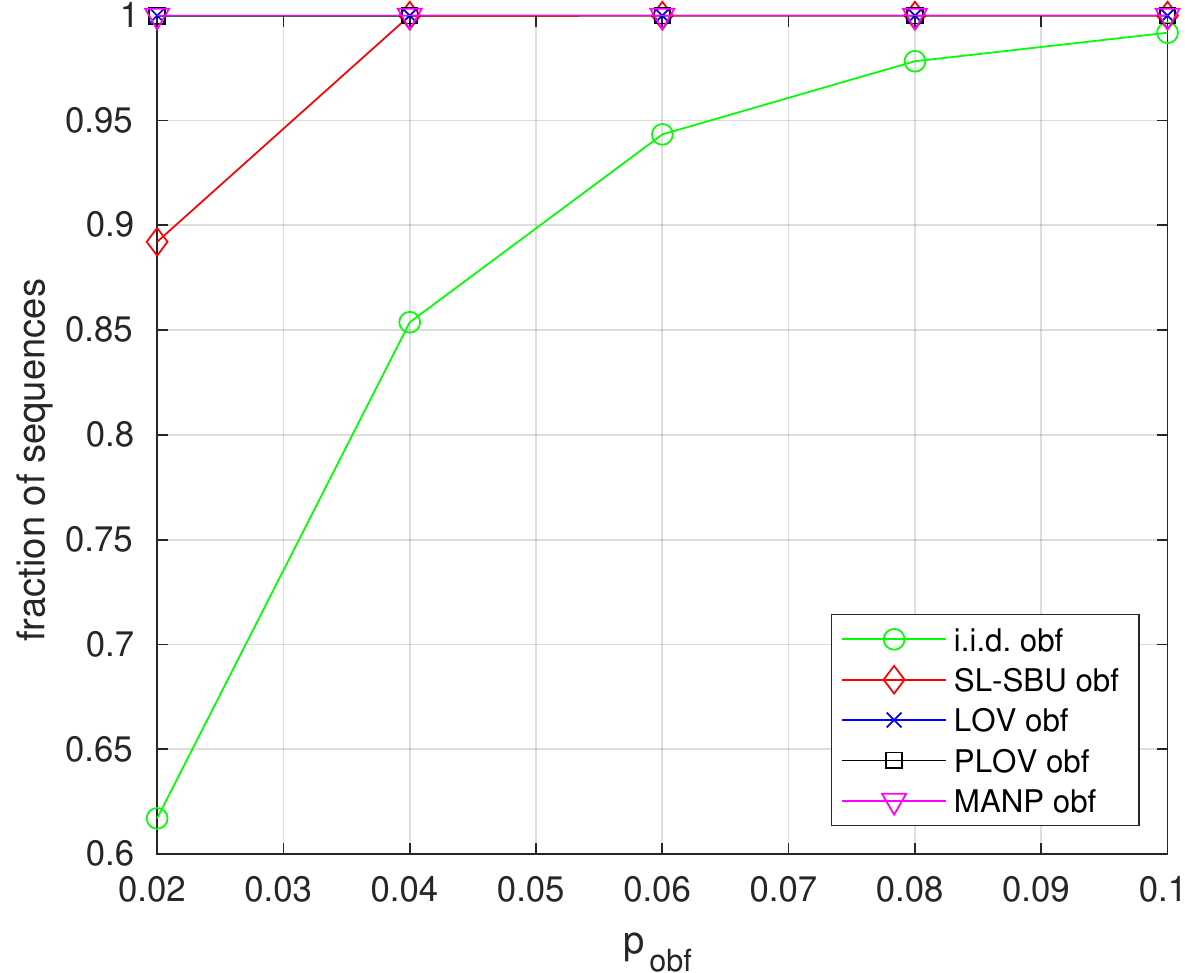}
         \caption{$p_{\text{obf}}=0.02:0.02:0.1$}
         \label{fig:perform_l1_Pobf_0dot02_0dot1_iid}
     \end{subfigure}
     \hfill
     \begin{subfigure}[b]{0.45\textwidth}
         \centering
         \includegraphics[width=\textwidth]{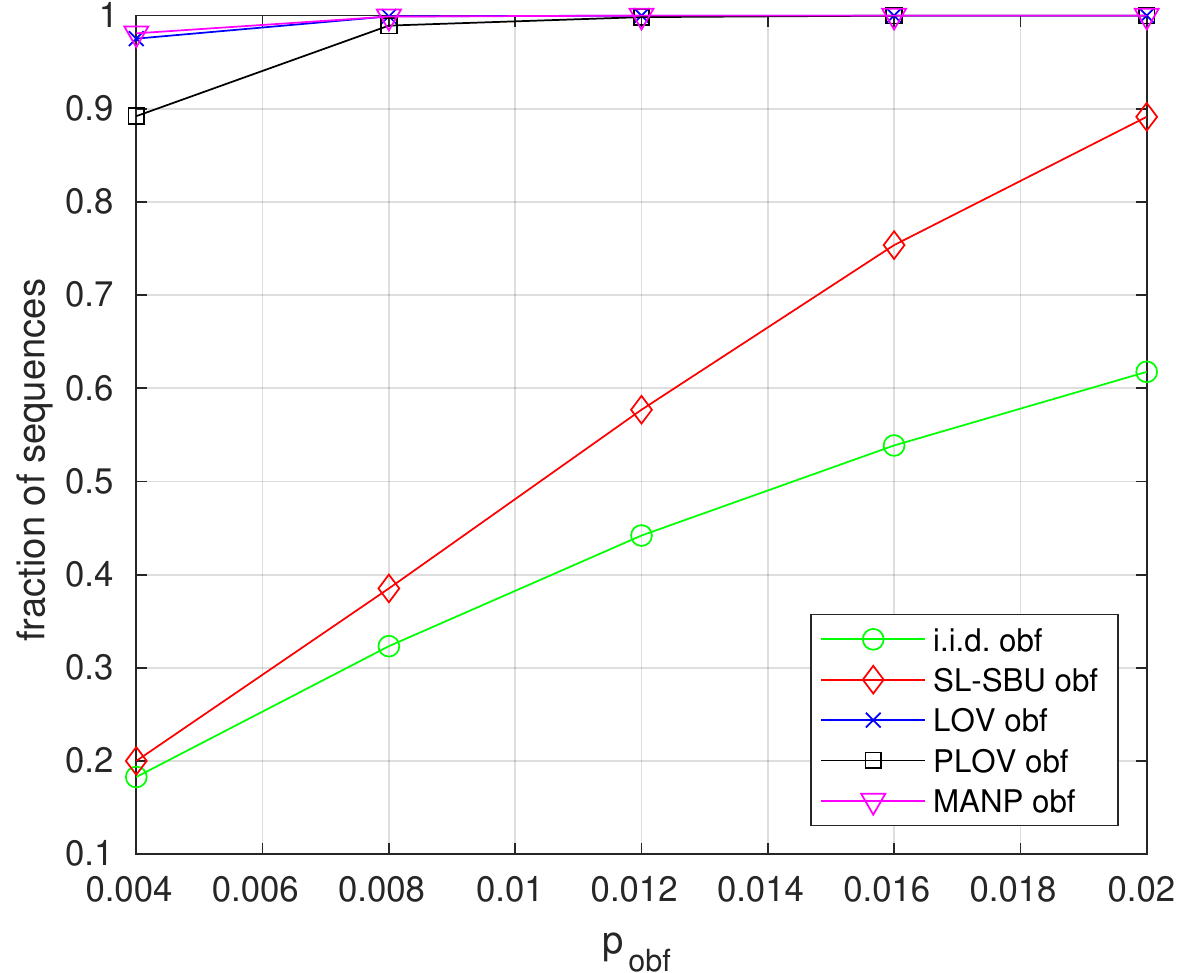}
         \caption{$p_{\text{obf}}=0.004:0.004:0.02$}
         \label{fig:perform_l1_Pobf_0dot004_0dot02_iid}
     \end{subfigure}
        \caption{Performance comparison of data-dependent obfuscation (DDO) methods (LOV, PLOV, MANP) and data-independent obfuscation methods (SL-SBU and i.i.d. obfuscation sequences) on i.i.d. sequences: the fraction of sequences which contain user $1$'s identifying pattern ($[r-l+1,\ldots,r-1,r]$). $r=20+l$, $l=1$, $m=1000$, $h=10$. Since here $l=1$, $h$ is only used for the MANP algorithm for creating the patterns within the distance requirement.}
        \label{fig:perform_l1_Pobf_iid}
\end{figure}

\begin{figure}
     \centering
     \begin{subfigure}[b]{0.45\textwidth}
         \centering
         \includegraphics[width=\textwidth]{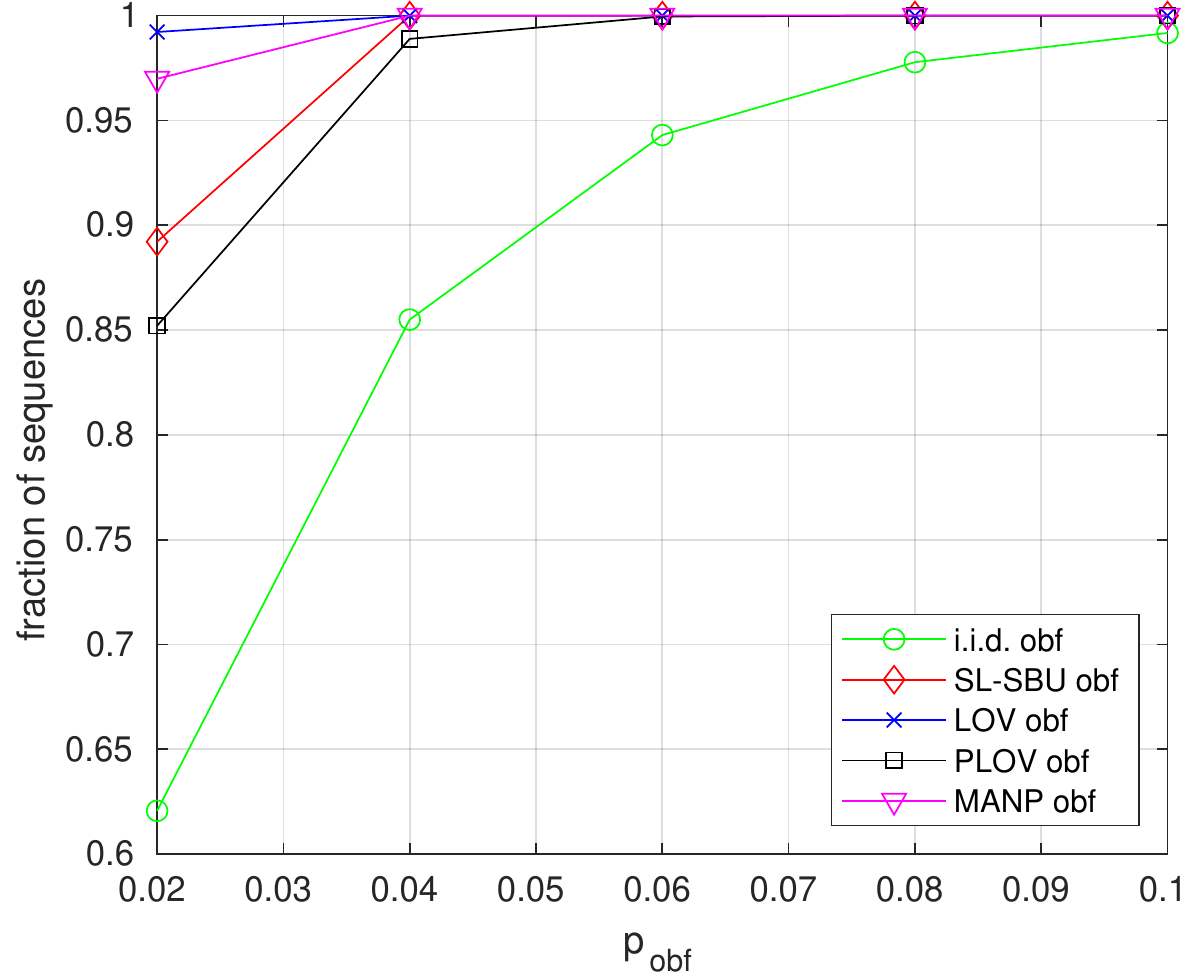}
         \caption{$p_{\text{obf}}=0.02:0.02:0.1$}
         \label{fig:perform_l1_Pobf_0dot02_0dot1_RealityMining}
     \end{subfigure}
     \hfill
     \begin{subfigure}[b]{0.45\textwidth}
         \centering
         \includegraphics[width=\textwidth]{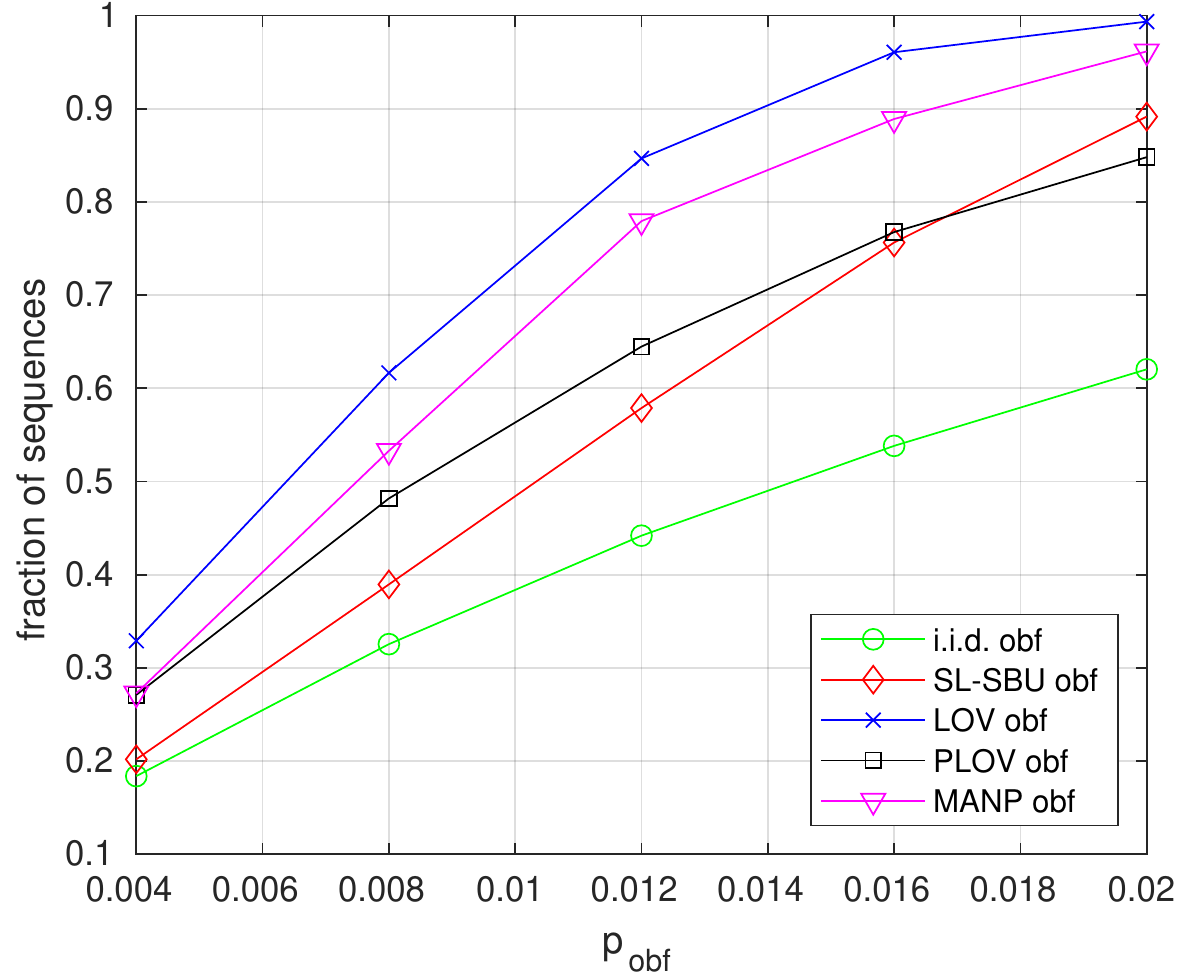}
         \caption{$p_{\text{obf}}=0.004:0.004:0.02$}
         \label{fig:perform_l1_Pobf_0dot004_0dot02_RealityMining}
     \end{subfigure}
        \caption{Performance comparison of data-dependent obfuscation (DDO) methods (LOV, PLOV, MANP) and data-independent obfuscation methods (SL-SBU and i.i.d. obfuscation sequences) on the Reality Mining dataset: the fraction of sequences which contain user $1$'s identifying pattern ($[r-l+1,\ldots,r-1,r]$). $r=20+l$, $l=1$, $m=1000$, $h=10$. Here $h$ is only used for the MANP algorithm for creating the patterns within the distance requirement. Data points are sampled with interval at least $10$ minutes.}
        \label{fig:perform_l1_Pobf_RealityMining}
\end{figure}

\begin{figure}
     \centering
     \begin{subfigure}[b]{0.45\textwidth}
         \centering
         \includegraphics[width=\textwidth]{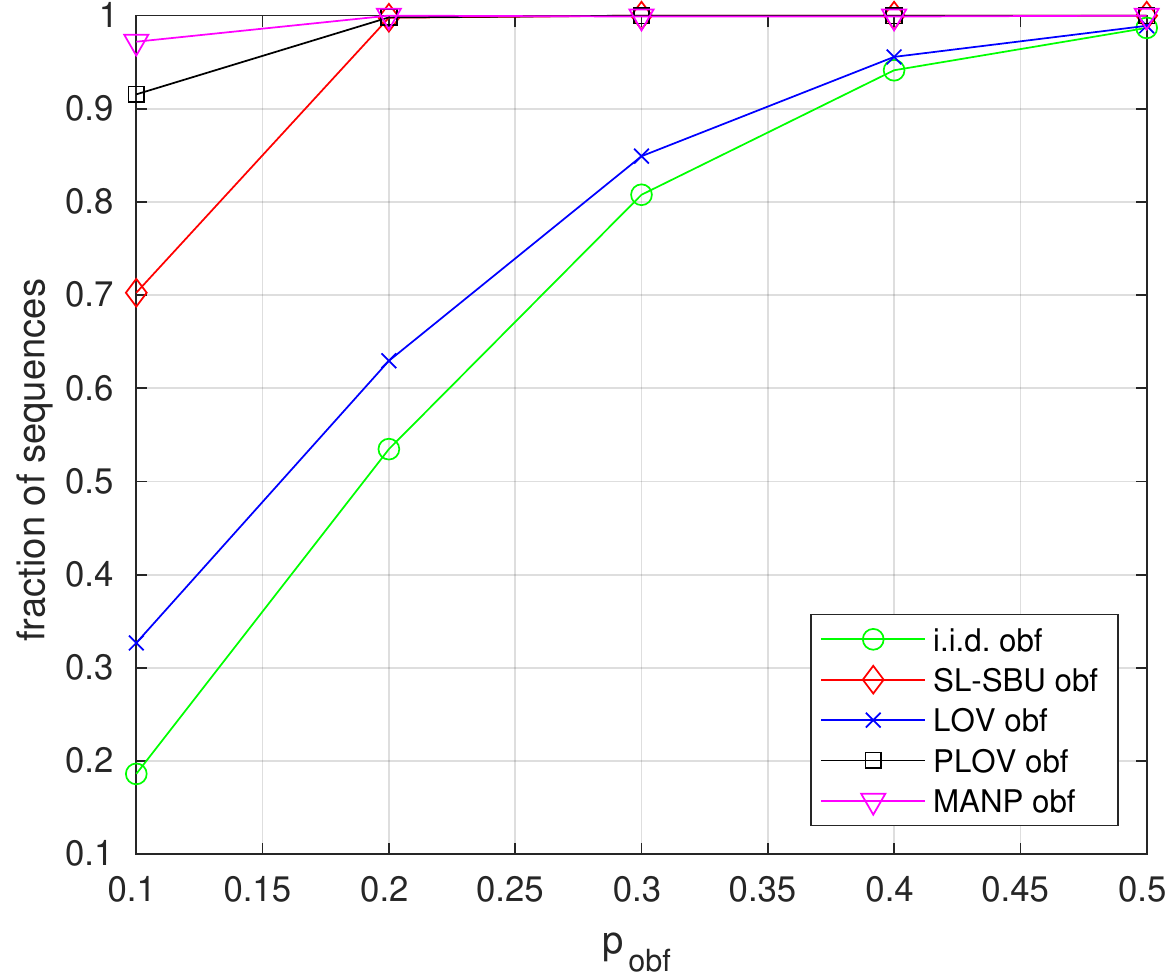}
         \caption{$p_{\text{obf}}=0.1:0.1:0.5$}
         \label{fig:perform_l2_Pobf_0dot1_0dot5_iid}
     \end{subfigure}
     \hfill
     \begin{subfigure}[b]{0.45\textwidth}
         \centering
         \includegraphics[width=\textwidth]{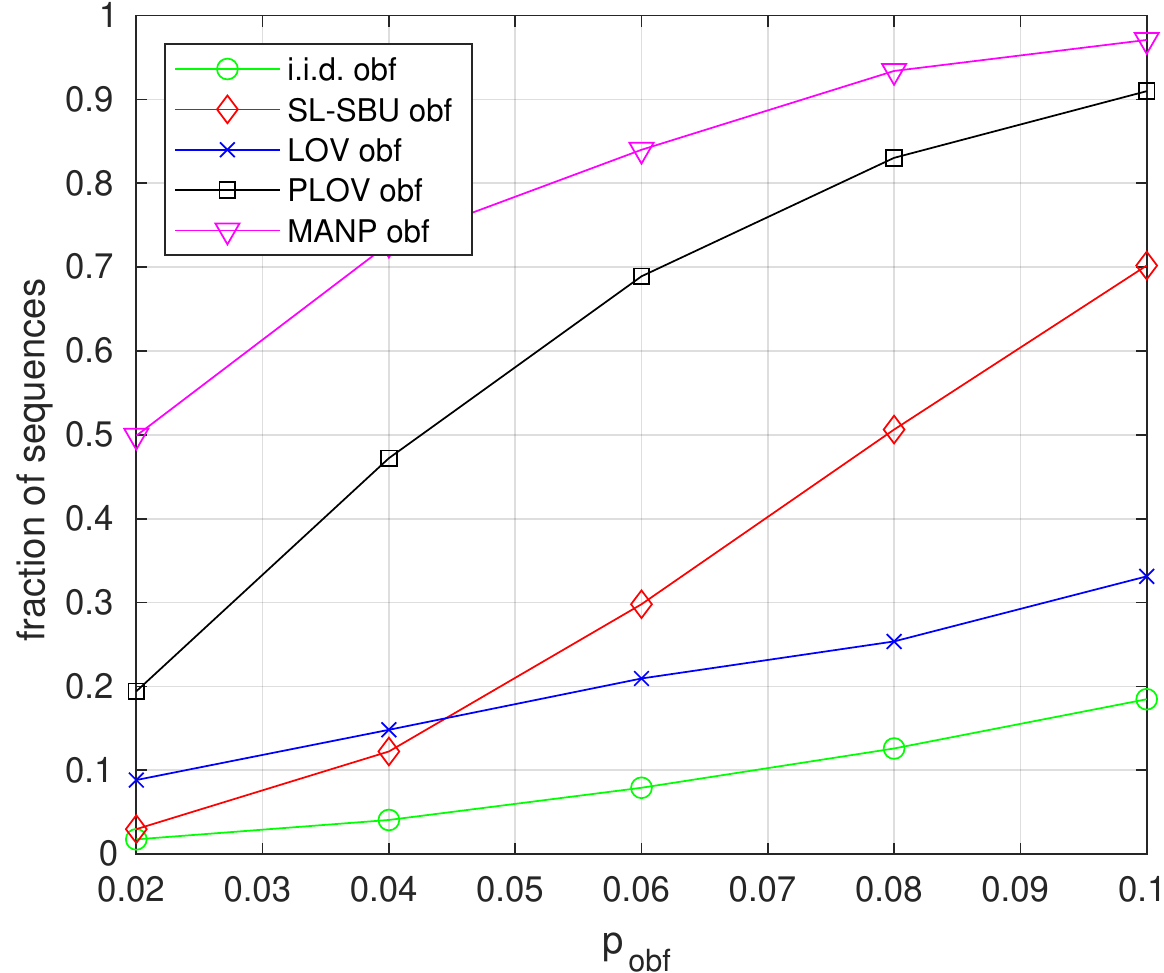}
         \caption{$p_{\text{obf}}=0.02:0.02:0.1$}
         \label{fig:perform_l2_Pobf_0dot02_0dot1_iid}
     \end{subfigure}
        \caption{Performance comparison of data-dependent obfuscation (DDO) methods (LOV, PLOV, MANP) and data-independent obfuscation methods (SL-SBU and i.i.d. obfuscation sequences) on i.i.d. sequences: the fraction of sequences which contain user $1$'s identifying pattern ($[r-l+1,\ldots,r-1,r]$). $r=20+l$, $l=2$, $m=1000$, $h=10$.}
        \label{fig:perform_l2_Pobf_iid}
\end{figure}

\begin{figure}
     \centering
     \begin{subfigure}[b]{0.45\textwidth}
         \centering
         \includegraphics[width=\textwidth]{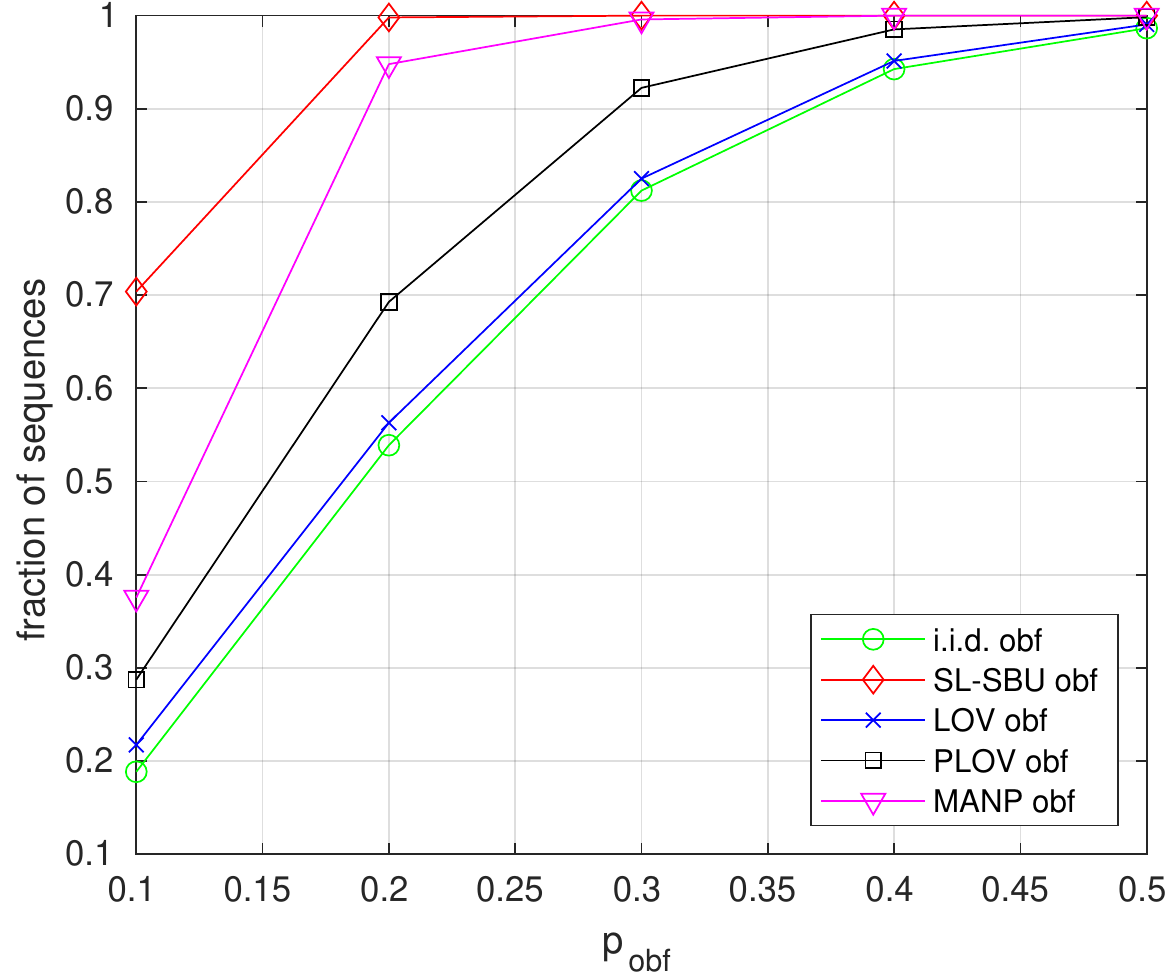}
         \caption{$p_{\text{obf}}=0.1:0.1:0.5$}
         \label{fig:perform_l2_Pobf_0dot1_0dot5_RealityMining}
     \end{subfigure}
     \hfill
     \begin{subfigure}[b]{0.45\textwidth}
         \centering
         \includegraphics[width=\textwidth]{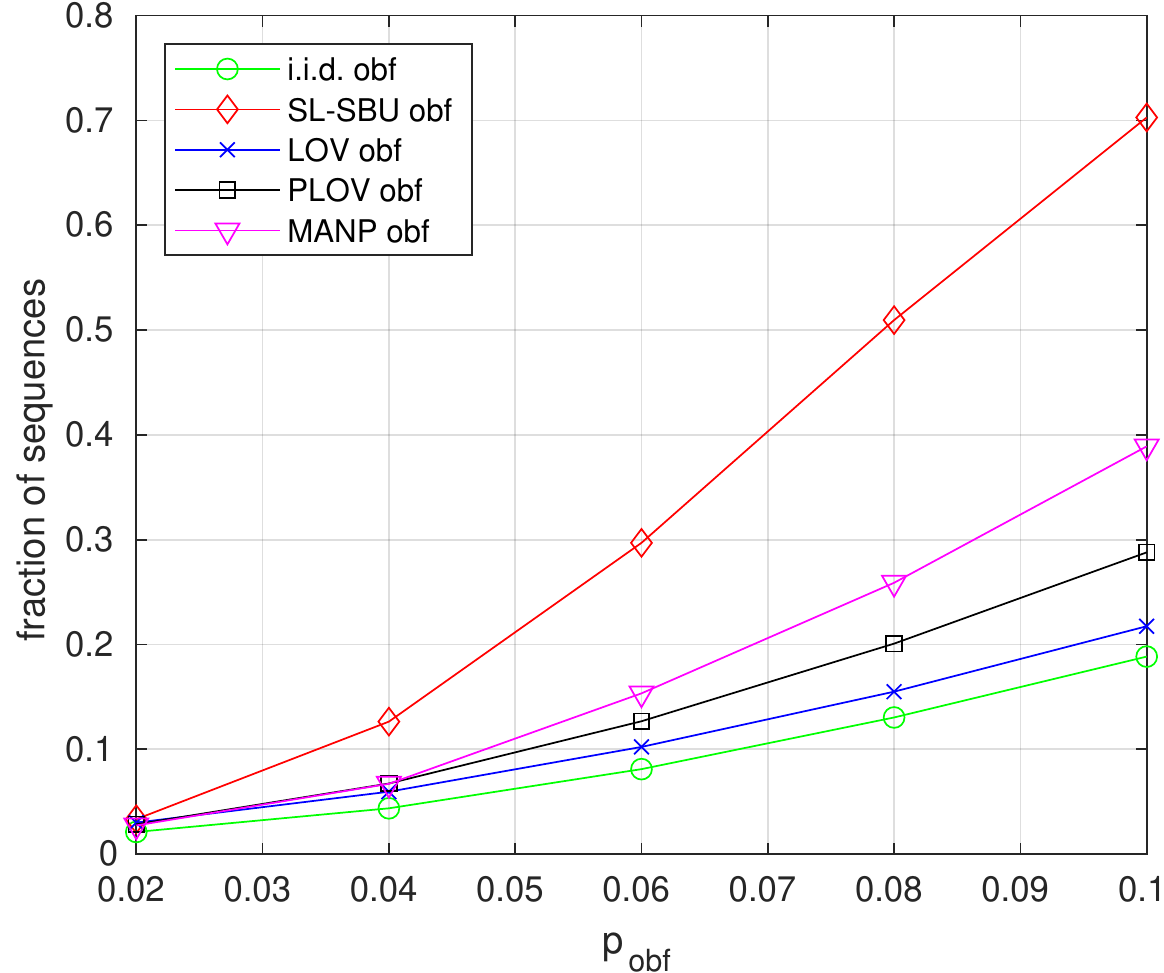}
         \caption{$p_{\text{obf}}=0.02:0.02:0.1$}
         \label{fig:perform_l2_Pobf_0dot02_0dot1_RealityMining}
     \end{subfigure}
        \caption{Performance comparison of data-dependent obfuscation (DDO) methods (LOV, PLOV, MANP) and data-independent obfuscation methods (SL-SBU and i.i.d. obfuscation sequences) on the Reality Mining dataset: the fraction of sequences which contain user $1$'s identifying pattern ($[r-l+1,\ldots,r-1,r]$). $r=20+l$, $l=2$, $m=1000$, $h=10$. Data points are sampled with interval at least $10$ minutes.}
        \label{fig:perform_l2_Pobf_RealityMining}
\end{figure}

\begin{figure}
     \centering
     \begin{subfigure}[b]{0.45\textwidth}
         \centering
         \includegraphics[width=\textwidth]{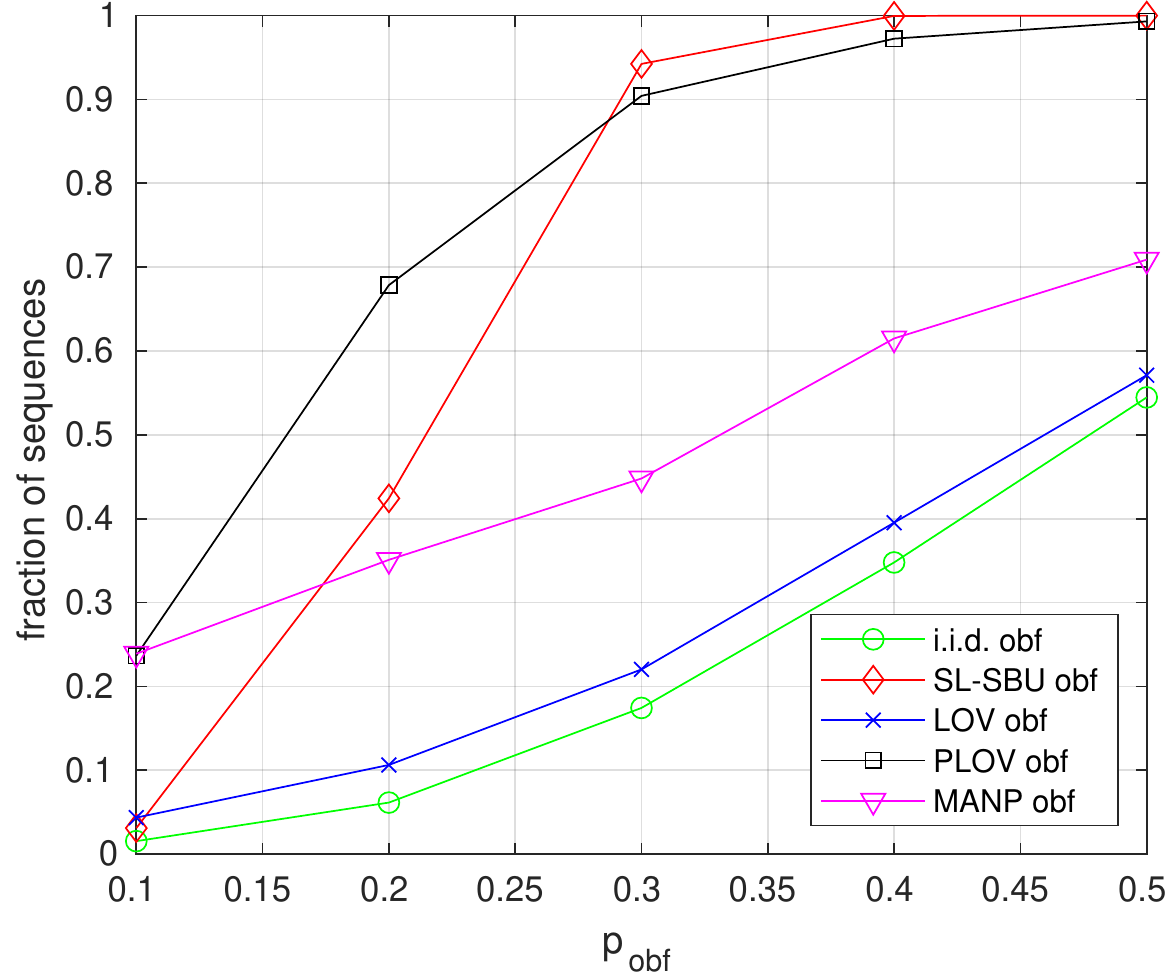}
         \caption{Performance on i.i.d. sequences}
         \label{fig:perform_l3_Pobf_0dot1_0dot5_iid}
     \end{subfigure}
     \hfill
     \begin{subfigure}[b]{0.45\textwidth}
         \centering
         \includegraphics[width=\textwidth]{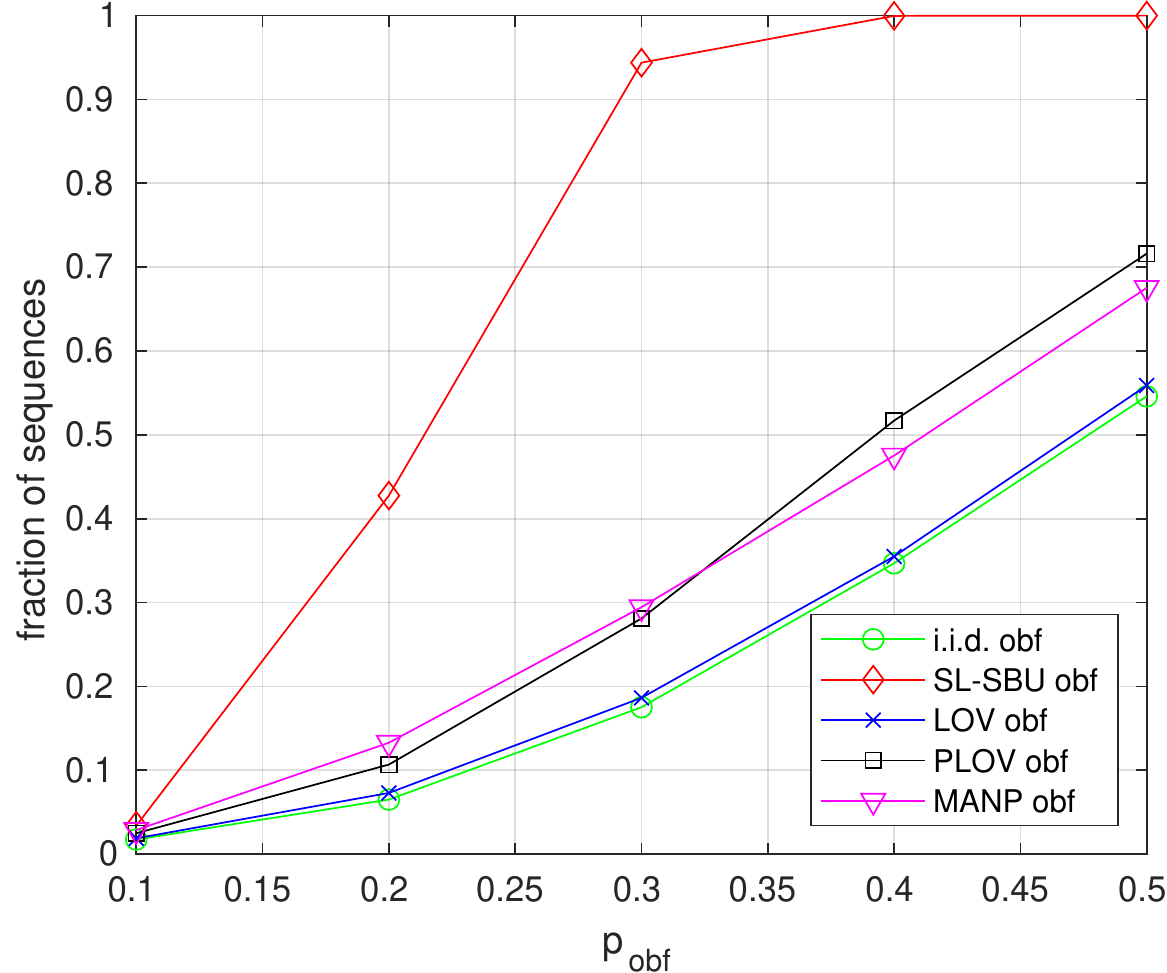}
         \caption{Performance on Reality Mining sequences}
         \label{fig:perform_l3_Pobf_0dot1_0dot5_RealityMining}
     \end{subfigure}
        \caption{Performance comparison of data-dependent obfuscation (DDO) methods (LOV, PLOV, MANP) and data-independent obfuscation methods (SL-SBU and i.i.d. obfuscation sequences) on i.i.d. sequences and Reality Mining sequences (data points are sampled with interval at least $10$ minutes): the fraction of sequences which contain user $1$'s identifying pattern ($[r-l+1,\ldots,r-1,r]$). $r=20+l$, $l=3$, $m=1000$, $h=10$.}
        \label{fig:perform_l3_Pobf}
\end{figure}

%For the De Bruijn \textit{superstring} sequence, its advantage is that it shows all the subsequences (patterns) as long as we have long enough sequence (at least $m \geq r^l  + l - 1$ to cover the complete \textit{superstring}), but we also require the distance between each pair of neighboring obfuscated letters for combining the pattern not far away from each other (distance less than or equal to $h$), otherwise, we will miss this potential hitting opportunity until the pattern shows again in the \textit{superstring}. For the randomness obfuscation sequence, its advantage is that the different obfuscated letters are flexibly generated at each place, and any location has the potential that it will hit the pattern, though it is pretty low as the location set becomes larger or the sequence length becomes longer.

%% file: conclusion.tex
\section{Conclusion}
\label{conclusion}
Various privacy-preserving mechanisms (PPMs) have been proposed to improve users' privacy in User Data-Driven (UDD) services.  To thwart pattern-matching attacks, we present data-independent and data-dependent PPMs that do not depend on a statistical model of users' data. In particular, a small noise is added to users' data in a way that the obfuscated data sequences are likely to have a large number of potential patterns; thus, for any user and for any potential pattern that the adversary might have to identify that user, we have shown that there will be a large number of other users with the same data pattern in their obfuscated data sequences.  We validate the proposed methods on both synthetic data and the Reality Mining dataset to demonstrate their utility and compare their performance.